\documentclass[journal]{IEEEtran}

\usepackage{fancyhdr}
\usepackage{hyperref}

\usepackage{amssymb}
\usepackage{amsmath}
\usepackage{amsfonts}
\usepackage{subfig}

\usepackage{bbm}
\usepackage[numbers]{natbib}
\usepackage{graphicx}
\usepackage{epstopdf}
\usepackage{color}
\usepackage{amsthm}
\usepackage{cancel}
\usepackage{caption}

\usepackage{algorithm,algorithmic}

\newtheorem{thm}{Theorem}[section]

\newtheorem{lemma1}{Theorem}[section]
\newtheorem{lemma}[lemma1]{Lemma}

\newtheorem{definition}{Definition}
\newtheorem{assumption}{Assumption}

\begin{document}
%
\title{Optimal Resource Allocation Over Time and Degree Classes for Maximizing Information Dissemination in Social Networks}

\author{Kundan Kandhway and Joy Kuri
\thanks{This is the complete version of the paper including supplementary material.}%
\thanks{Authors are with the department of Electronic Systems Engineering, Indian Institute of Science, Bangalore 560012, India. Email: \{kundan, kuri\}@dese.iisc.ernet.in}
}

\maketitle

\thispagestyle{fancy}

\begin{abstract}
We study the optimal control problem of allocating campaigning resources over the campaign duration and degree classes in a social network. Information diffusion is modeled as a Susceptible-Infected epidemic and direct recruitment of susceptible nodes to the infected (informed) class is used as a strategy to accelerate the spread of information. We formulate an optimal control problem for optimizing a net reward function, a linear combination of the reward due to information spread and cost due to application of controls. The time varying resource allocation and seeds for the epidemic are jointly optimized. A problem variation includes a fixed budget constraint. We prove the existence of a solution for the optimal control problem, provide conditions for uniqueness of the solution, and prove some structural results for the controls (e.g. controls are non-increasing functions of time). The solution technique uses Pontryagin's Maximum Principle and the forward-backward sweep algorithm (and its modifications) for numerical computations. Our formulations lead to large optimality systems with up to about 200 differential equations and allow us to study the effect of network topology (Erd\H os-R\'enyi/scale-free) on the controls. Results reveal that the allocation of campaigning resources to various degree classes depends not only on the network topology but also on system parameters such as cost/abundance of resources. The optimal strategies lead to significant gains over heuristic strategies for various model parameters. Our modeling approach assumes uncorrelated network, however, we find the approach useful for real networks as well. This work is useful in product advertising, political and crowdfunding campaigns in social networks.
\end{abstract}


\IEEEpeerreviewmaketitle

\section{Introduction}

\IEEEPARstart{M}{aximizing} the reach of a piece of information is of interest to many entities, such as, political parties (during elections), companies marketing new products, governments and NGOs (to spread awareness about some socially relevant issue), etc. Before the advent of online social networks, information used to disseminate in a population (only) due to social contacts between individuals in day-to-day life and through the mass media. With the explosion in the number of people using online social networks these days---which has added to the available channels over which information travels---the extent and speed of information propagation is higher than ever. Campaigners are leveraging these social networks in an attempt to maximize the reach of their messages, as illustrated by the successful use of social networks such as Twitter and Facebook in the 2012 US Presidential elections \cite{rutledge2013obama} and 2014 Indian Parliamentary elections \cite{theeconomictimes2014facebook}.

This paper studies campaigning strategies aimed at maximizing the spread of information in a population for a fixed-duration political, advertisement or social-awareness campaigns. Specifically, the campaigning resource---such as money, manpower or logistics---is allocated optimally over time (the campaign duration), and classes of individuals carved up based on their degrees. The degree of an individual is the number of connections she has with others in the network. We use only node degree information, which makes this work useful for both partially observed networks---where the exact connection pattern is not completely known, but only node degree has been estimated---an example of which is the face-to-face human interaction network; and fully observed online social networks; or a combination of both.

Biological epidemic models are used to model information spread in a population due to similarities in the ways a communicable disease and information spread over a network (\emph{e.g.} \cite{karnik2012, kandhway2014run}). We have modeled information diffusion as a Susceptible-Infected (SI) epidemic process. Earlier works (such as \cite{banerjee2014epidemic}) have favored the SI process as a model for information diffusion. The SI model divides the population into two compartments. A susceptible node is yet to receive the message and an infected (informed) node has received and is spreading the message. Most nodes are susceptible at the beginning of the campaign---except a small fraction of the population which acts as the seed for the epidemic---and changes to the infected state due to interaction with infected neighbor(s). Once the node becomes infected, it stays in that state. The SI process is suitable for modeling situations when individuals receiving the message do not forget it. This happens for campaigns of short duration (\emph{e.g.} those for charities) or, situations such as political campaigns which generate a lot of interest among people. It is also suitable for marketing of long lasting products such as video games and smart phones.

Other models which may be used, depending on the situation, are: (a) Susceptible-Infected-Susceptible, in which nodes alternate between susceptible and infected state, suitable for marketing consumable goods with substitutable brands, (b) Susceptible-Infected-Recovered, in which nodes lose interest in spreading after fixed amount of time since being infected, (c) Maki-Thompson model, in which nodes lose interest in spreading after meeting a fixed number of informed individuals, etc.

We adapt the standard SI model to include \emph{time varying direct recruitment of nodes from susceptible to infected class} to accelerate information diffusion. Similar problems were addressed in \cite{dayama2012optimal,banerjee2014epidemic}. Such a control may be implemented by emailing or texting individuals, posting messages on their social network timelines, or by placing advertisements in the mass media. Resource and manpower constraints will prevent the advertiser/campaigner from communicating to all individuals in the network. Hence, to maximize the extent of information dissemination, it is important to identify types of individuals who should be targeted and times when campaigning should run with stronger intensity. For example, should we target high degree nodes because they act as hubs and may be better spreaders in the network? Or should we target low degree nodes who are at a disadvantage in receiving the messages from others by epidemic spreading due to fewer connections, and let more connected higher degree nodes receive the message from epidemic spreading. As shown by our results, the network topology and the amount of resource available and/or the cost of application of the control affect the answer to this question.

\subsection{Related Work and Our Contribution}

Many authors have addressed disease mitigation problems on heterogeneous networks \cite{borgs2010distribute, preciado2013optimal}; however, optimal control of heterogeneous networks has received less attention. Optimal control has found use in devising strategies for preventing the spread of disease epidemics and computer viruses in works such as \cite{asano2008optimal, zhu2012optimal}. Our work differentiates itself from these by considering a \emph{population of networked individuals instead of homogeneous mixing.} Considering networked individuals and many controls, as is the case here, leads to a huge optimality system with about 200 differential equations in our case, compared to only a few equations in the above works. Also we discuss the case of \emph{fixed budget constraint}, which necessitates modifications to the standard forward-backward sweep algorithm needed to solve the optimal control problem. Networks play an important role in epidemic spreading (of both information or disease) because people interact with and trust only a small subset of the total population to which they are `connected'. Homogeneous mixing assumes that any node is equally likely to meet any other node in the population and is quite farfetched.

Our method captures information dissemination dynamics more accurately than the homogeneous mixing models in the prior works which study optimal control of information epidemics \cite{karnik2012, kandhway2014run, kandhway2014optimal} (and security patches in the case of \cite{khouzani2011optimal}) and provide more accurate controls (campaigning strategies). Also, we are able to discuss the \emph{influence of node degrees on resource allocation}, which is not possible when homogeneous mixing is considered.

Information dissemination with impulsive controls in a homogeneously mixed population was considered in \cite{belen2008}. However, most systems can be controlled throughout the campaign horizon---\emph{e.g.} advertisements appear regularly in an individual's social network timeline or in the mass media---and not just once or twice. This motivates a model which allows for resource allocation \emph{throughout} the campaign horizon, as is the case in our work, and not just at a few time instants as in \cite{belen2008}.

The work in \cite{dayama2012optimal} has considered heterogeneous mixing of the population in devising optimal strategies for product marketing, but the results were presented by dividing the population into two degree classes only; this is inadequate to study the \emph{effect of node inter-connection topologies (scale-free/Erd\H os-R\'enyi degree distribution) on controls}. To achieve this, we consider a population with up to 100 degree classes. Note that real social networks follow a scale-free degree distribution \cite{newman2009networks} and two classes are not adequate to characterize them. The authors in \cite{youssef2013mitigation} formulated a problem to mitigate a biological epidemic on a network. Our approach uses one differential equation for each degree class as against one for each node, as was the case in \cite{youssef2013mitigation}, and thus scales with number of degree classes in the network and not the actual population size as in \cite{youssef2013mitigation}. Also, that work presented optimal results for only a simple case of a five-node network and proposed heuristics for larger networks.

Notice that epidemic models used in both \cite{dayama2012optimal, youssef2013mitigation} work well on uncorrelated networks, which is also the case in this work. Thus, our modeling approach is no worse than in these works and in addition, we are able to study the effect of node inter-connection topology on time-varying controls due to the large system size considered in this work. In uncorrelated networks, connections are constrained by degree distribution of the network, but are otherwise completely random. In contrast, real social networks show some level of connection correlations, \emph{e.g.}, friends of an individual are also likely to be friends. Unlike the previous works, we have tested this modeling approach on a real network via simulations for both uncontrolled and controlled cases. The results show that the control strategies derived from our model lead to improvements with respect to heuristic strategies on real networks as well.

Optimal seed selection for maximizing the influence in a network with known connections was studied in \cite{kemp2003maximizing}. However, once the seed is decided, the process evolves in an uncontrolled manner. In contrast, we allow for the information diffusion process to be controlled throughout the campaign horizon and present results for a problem which allows for \emph{joint seed selection and resource allocation over time} to maximize information diffusion.

The works in \cite{banerjee2014epidemic, banerjee2014epidemicthreshold} compute bounds on the spreading times and epidemic thresholds for SI/SIS epidemic influenced by external agents using tools from the probability and graph theories. Although these works show order optimality of the uniform spreading  strategy with respect to strategies which can be tailored to network state for specific networks (such as line/ring, grid and spatially constrained random geometric graphs), computation of optimal strategies under a cost criterion was not undertaken in these works. These networks do not have heavy tails and long distance links as observed in (small world) real social networks to which such conclusions may not be generalized. In contrast, we compute optimal strategies for \textit{scale-free and real networks} which provide more accurate insights to campaigners. If a cost criterion is included in problems formulated with exact network state, the complexity of computing the solution will increase exponentially with network size. A mean field approach for modeling the SI epidemic provides computational tractability for handling large networks---because computational complexity grows only as the number of degrees in the network and not the network size.

The authors in \cite{sheldon2010maximizing} formulate a mixed integer program to maximize the spread of cascades in a network. However, the intervention involves adding new nodes and edges in the network: an approach different from ours.

We list the \textbf{\emph{main contributions of this paper}}. We begin by adapting the standard SI model to include a time-varying control which recruits individuals from susceptible to infected class to accelerate information dissemination. We define a net reward function which is a weighted combination of reward due to the extent of information dissemination in the population and cost due to application of control, and formulate a problem to maximize the net reward. The problem jointly optimizes the seeds of the epidemic and time varying controls to maximize the reward. We also study the fixed budget variation of this problem (where seeds are given and not optimization variables). To the best of our knowledge, the joint optimization problem formulation does not exist in the literature.

We show the existence of a solution to our problem using Cesari's theorem and provide some structural results for the shapes of the controls. These results seem novel in the context of a networked population and a controlled SI model. Further, we provide a sufficient condition under which the solution to the optimal control problem is unique. To solve the optimal control and joint seed optimization and control problems, we propose numerical algorithms which make use of Pontryagin's Maximum Principle. The standard forward-backward sweep algorithm needs to be adapted to take care of the specific formulations in this paper (\emph{e.g.} fixed budget constraint). We study the convergence of the forward-backward sweep algorithm for our system. Our formulation requires solving a large number of control functions ($\approx 100$) and leads to large dynamical systems with a similar number of differential equations. For the joint (seed-control) problem, the same number of optimization variables also need to be optimized in addition to finding the optimal control functions.

We quantify the improvement achieved by the optimal strategies over simple heuristics for scale-free and Erd\H os-R\'enyi configuration model networks. In many cases, large improvements are observed. Results also reveal that the resource allocated to the degree classes under the optimal strategy changes with system parameters and network topology. For example, in the case of scale-free networks, when less resource is available, the optimal strategy allocates more resource to high degree classes but for the abundant resource case, low degrees get more resources than medium degrees. In the case of Erd\H os-R\'enyi networks, even for the scarce resource case, medium degrees get the most allocation.

Our information diffusion model assumes that network connections are uncorrelated; thus, may not be accurate for real social network. We test the accuracy of our modeling approach on a real social network via simulations and find the model to work well even for real social networks.

\section{System Model and Problem Formulation}
\label{sec:sys_model_prob_formaulation}

Individuals in the population are organized in a social network (graph). The network is undirected and remains static over the duration of the campaign. The number of other nodes/individuals a given node is connected to is termed as the degree of that node in the graph. A node with degree $k$ is said to be in degree class $k$. The set of all degree classes in the graph under consideration is represented by $\mathbb K=\{k:K_{min}\leq k\leq K_{max}\}$, for two positive integers $K_{min}$ and $K_{max}$. The network is characterized by its `degree distribution', $p_k$, which is the probability that a randomly chosen node in the social network belongs to degree class $k\in\mathbb K$. One can empirically calculate $p_k=N_k/N$, where $N_k$ is the number of nodes in degree class $k$ and $N=\sum_{k\in\mathbb K}N_k$. We first explain the uncontrolled model and then adapt it to include the controls to formulate the optimal control problem.

\subsection{Uncontrolled SI Epidemic on a Network}
We model the SI process using the `degree based compartmental model' \cite{may1988transmission}. It works best on networks which lack any correlations in the degrees of two neighbors. More precisely, a half edge (when an edge is cut, it leads to one half edge each at the two neighboring nodes the edge connected) from any node is equally likely to connect to any other half edge in the network. Such networks are called `configuration model networks'.

The degree based compartmental model assumes that all nodes in degree class $k$ have the same statistical behavior in the network. That is, any node with degree $k$ in the network has the \emph{same} probability of being in infected (or informed) state at any time $t$. In reality, a node in a dense core is more likely to be infected than a node at the periphery of the network; however, if the variance of the probability distribution of being in the infected state is low, this approximation works well. This happens for configuration model networks in the limit of large network size, $N\rightarrow\infty$. Thus the degree based compartmental model may be termed as a `mean field model'.

Nodes in the network lie in either of the two states---susceptible or infected (informed). An infected individual is aware of the message and is spreading it to her susceptible neighbors who are yet to receive the message. The campaign runs during $t\in[0,T]$, where $T$ is termed as the campaign deadline. At any time instant $t$, the fractions of susceptible and infected nodes in the degree class $k$ are denoted by $s_k(t)$ and $i_k(t)$. Note that $s_k(t)=1-i_k(t)$ is not an independent state variable of the system. The total fractions of susceptible and infected individuals in the network at time $t$ is given by $s(t)=\sum_{k\in\mathbb K}p_ks_k(t)$ and $i(t)=\sum_{k\in\mathbb K}p_ki_k(t)$. Again, $s(t)=1-i(t)$.

The information epidemic is characterized by its spreading rate profile $\beta(t)\geq 0,~t\in[0,T]$. We have allowed the spreading rate to vary over time because the \emph{interest of the target population in the subject of the campaign may vary with time.} For example, one can observe monotonically increasing interest as the election day approaches, monotonically decreasing interest in a product as it becomes old after its release or fluctuating interest in movie tickets where demand is more during weekends and less during weekdays. A single susceptible-infected contact passes the information from the infected to the susceptible node with a probability $\beta(t)dt$ at time $t$, in a small interval $dt$. The epidemic starts with a small fraction of infected nodes at $t=0$ (also called the seed nodes) and spreads stochastically in the social network. We assume that $i_k(0)=i_{0k},~\forall k\in\mathbb K$ act as seeds in degree class $k$, $0\leq i_{0k} \leq 1$.

We now discuss briefly the notions of `neighbor degree distribution' and `excess degree distribution' which will be used later in this section. The neighbor degree distribution, $r_k$, is the probability that we will reach a neighboring node of degree $k$ by following an edge of any node in the network. For configuration model networks, $r_k=kp_k/\bar k$, where $\bar k\triangleq \sum_{k\in\mathbb K}kp_k$ is the mean degree of the graph \cite[Sec. 17.10.2]{newman2009networks}. It is biased towards higher degrees. Such behavior is expected because high degree nodes are more connected and will be reached more often (\emph{e.g.} there is no way to reach nodes with degree 0 by following an edge, so $r_0=0$, even if $p_0\neq 0$).

Denote by $M_h$ (even), the total number of half edges in a network. The number of nodes with degree $k$ in the network $=Np_k$, and number of half edges originating at them $=kNp_k$. Consider a half edge at any given node in the network. The probability that it will be connected to a neighbor with degree $k$ is $r_k=kNp_k/(M_h-1)\approx kNp_k/M_h$ (for large networks). But $M_h/N=$ mean degree in the network; so $r_k=kp_k/\bar k$.

For epidemic spreading, the quantity of interest is excess degree distribution, $q_k\triangleq r_{k+1}=(k+1)p_{k+1}/\bar k$ \cite[Sec. 17.10.2]{newman2009networks}. Consider a susceptible node $A$ in the network. The neighbors of susceptible node $A$, if infected, could not have got the information from $A$. So we discount the edges from this susceptible node $A$ to its neighbors and the neighbors behave like nodes with degree $1$ less than their actual degrees.

The message can be passed to a given susceptible node of degree $k$ from its infected neighbors, whose (mean) number is given by $k\sum_{l\in\mathbb K}(q_li_l(t))$, where $q_k$ is the excess degree distribution discussed above. Assuming neighbors interact independently, the probability that the message is transferred to this susceptible node in an interval $dt$ at time $t$ = $1-$probability that none of the infected neighbors infects her = $1-(1-\beta(t)dt)^{k\sum_{l\in\mathbb K}(q_li_l(t))} \approx \beta(t)k\sum_{l\in\mathbb K}(q_li_l(t))dt$. The fraction of susceptible nodes in degree class $k$ is $s_k(t)$. Hence, the total increase in the fraction of infected nodes in degree class $k$ in an interval $dt$ at time $t$ is given by $s_k(t)\beta(t)k\sum_{l\in\mathbb K}(q_li_l(t))dt$. This leads to the rate of change of the fraction of infected individuals in degree class $k$ in an uncontrolled SI epidemic \cite[Sec. 17.10.2]{newman2009networks}:
\begin{align}
\frac{d}{dt} i_k(t) = \beta(t)ks_k(t)\sum_{l\in\mathbb K}(q_li_l(t)),~k\in\mathbb K. \label{eq:uncontrolled_SI}
\end{align}

\subsection{Controlled SI Epidemic on a Network}

To aid information dissemination, we introduce a control signal $u_k(t)$ in each degree class $k\in\mathbb K$. The control recruits susceptible individuals in the respective degree class and converts them into infected nodes in degree class $k$. In political/product marketing campaigns, $u_k(t)$ represents the rate at which attempts are made to recruit individuals in degree class $k$. This can be done by posting messages in their Facebook/Twitter time-line for example, and requesting them to re-post/re-tweet the message to their contacts. If the targeted node was not aware of the message, we achieve a recruitment. In the case of a company launching a new medicine, $u_k(t)$ may represent the rate at which medical representatives visit doctors at time $t$. We define the set of all admissible control functions in the following:

\begin{definition}[Set of all admissible controls]
\label{def:set_of_admissible_controls}
Define, $U \triangleq \{u:u \textnormal{ is Lebesgue measurable on } t\in [0,T]\textnormal{ and } u(t)\in \mathbb R\}$. Then, the set of all admissible control functions is given by: $U^{|\mathbb K|}=\Big\{\boldsymbol u=\{u_k,~k\in\mathbb K\}:u_k\in U\Big\}.$
\end{definition}

In the following, we formulate the joint optimization--optimal control problem and then discuss the formulation.
\begin{subequations}
\label{eq:opt_prob}
\begin{align}
\underset{ \begin{smallmatrix} \boldsymbol u, \big\{ \boldsymbol{i_0}:~0\leq i_{0k}\leq 1, \\ \sum\limits_{k\in\mathbb K} p_ki_{0k}=B_{i_0} \big\} \end{smallmatrix} }{\text{maximize}} J & = \sum_{k\in \mathbb K} p_k i_k(T)-\int_0^T\sum_{k\in \mathbb K}g_k(u_k(t))dt, \label{eq:cost_funtion}\\
\text{s.t.:~~} \frac{d}{dt} i_k(t) & =  \beta(t) ks_k(t) \sum_{l\in \mathbb K} \left(q_{l} i_l(t)\right) + \gamma(t)u_k(t)s_k(t);\nonumber \\
& ~~~~~~~~~~~~~~~~~~~~~~~~~~~~~~~~~k\in \mathbb K. \label{eq:opt_prob_states} \\
i_k(0) & =  i_{0k}; ~~k\in \mathbb K. \label{eq:opt_prob_init_cond}
\end{align}
\end{subequations}
In the above formulation, $s_k(t)=1-i_k(t)$, are not independent variables. Hence Problem (\ref{eq:opt_prob}) has only $|\mathbb K|$ state variables, which are represented by the $i_k$'s (and not 2$|\mathbb K|$). In many scenarios it may be possible to decide the initial set of infected nodes (\emph{e.g.} brand ambassadors recruited by companies) in addition to deciding resource allocation over time. The seed which kick starts the epidemic is given by the vector $\boldsymbol{i_{0}}=\{ i_{0k},~k\in\mathbb K \}$, where $0\leq i_{0k}\leq 1$ is the fraction of individuals selected as seed in the degree class $k$. Here $B_{i_0}$ is the `seed budget', the initial fraction of infected nodes in the whole network.

The evolution of the state variables $i_k(t)$ in the controlled system is governed by Eq. (\ref{eq:opt_prob_states}), with initial condition (\ref{eq:opt_prob_init_cond}). Notice the additional term $\gamma(t)u_k(t)s_k(t)$ in Eq. (\ref{eq:opt_prob_states}) compared to Eq. (\ref{eq:uncontrolled_SI}). The term incorporates effect of control in the $k$th degree class. Here $s_k(t)$ is multiplied by $\gamma(t)u_k(t)$ because the control is only effective on the susceptible individuals. Here $\gamma(t)\geq 0$ captures the effectiveness of the control $u_k(t)$ at time $t$. Since the interest of the population in the subject of the campaign varies with time, recruitment should be easier when the interest is more and vice versa.

The reward function is given in Eq. (\ref{eq:cost_funtion}). We have used a weighted combination of reward due to information dissemination (given by $\sum_{k\in\mathbb K}p_ki_k(T)$) and cost due to application of controls (given by $\int_0^T\sum_{k\in \mathbb K}g_k(u_k(t))dt$) as the net reward function (note that constant weights are subsumed by $g_k(.)$). We consider the instantaneous cost of application of control as a function of $\boldsymbol u(t)$, and have not multiplied it with $s_k(t)$ or $i_k(t)$. This is so because in most practical situations, the state of the node (susceptible/infected) is either unknown or there is a cost involved in determining it (\emph{e.g.}, we do not know if an individual already knows about a newly launched product). Hence, campaigning either does not consider the states of the specific nodes (information is passed to both classes), or the cost incurred is similar irrespective of node state (in infecting or determining state of the node, for example, by using text analytics on the time-line posts of the node). Similar assumptions were made in prior works such as \cite{karnik2012,kandhway2014run} (but for a homogeneously mixed population).

Also, we have used a reward function which only considers the final fraction of infected individuals at $t=T$, and not the system evolution over $0\leq t<T$. This is suitable for situations such as political campaigns, marketing of durable products/services (\emph{e.g.} video-games, cell-phone plans) etc., where the final number (fraction) of the infected population is the quantity of interest.

In the special case of Problem (\ref{eq:opt_prob}) where $\boldsymbol{i_0}$ is not an optimization variable but rather a given vector, we get an optimal control problem where only the control vector function $\boldsymbol{u}$ is to be optimized.

\begin{assumption}
\label{assumption:gk_increasing}
We assume all the cost functions, $g_k(u_k(t)),~k\in\mathbb K$, to be non-negative, monotonically increasing and strictly convex in their arguments in the region $u_k(t)\geq 0$. Further, $g_k(0)=0,~k\in\mathbb K$.\footnote{This is natural because cost will increase with the control strength in any practical situation. Also, the convexity assumption holds in many economic applications.}
\end{assumption}

For the situation we consider---maximizing the reach of useful information---the controls are non-negative. For some arbitrary cost function this may not be true. However, the following assumption on the cost function $g_k(u_k(t)),~k\in\mathbb K$, ensures that negative values of controls are not optimal.
\begin{assumption}
\label{assumption:gk_even}
We assume all the cost functions $g_k(u_k(t)),~k\in\mathbb K$, to be even functions.\footnote{This ensures that negative values of controls are not optimal. A negative control at time $t$ will incur the same cost as its modulus. However, a negative value reduces the value of reward $J$ in (\ref{eq:cost_funtion}), as it would take away individuals from the infected class (Eq. (\ref{eq:opt_prob_states})). Instead, if $|u_k(t)|$ is applied, the cost incurred is the same and the value of the reward is more than that in the case when control is negative. Hence for any time instant, negative values of controls will never be optimal and hence we do not need to add the additional constraint $u_k(t)\geq 0$ in Problem (\ref{eq:opt_prob}).}
\end{assumption}

The above assumption simplifies further analysis and numerical computation of controls. It is not restrictive in practical scenarios, because controls are never negative, and we can take even extensions of cost functions defined for $u_k(t)\geq 0$ as $g_k(.)$. Also notice that in Problem (\ref{eq:opt_prob}), we have not explicitly enforced the conditions that the fractions $i_k(t)$, $s_k(t)$ lie in [0,1]. This is due to the following:

\begin{lemma}
\label{thm:ik_sk_in_0_1}
Let $\psi_k(t)$ and $\eta_k(t)=1-\psi_k(t)$ be the solutions to the system of differential equations (\ref{eq:opt_prob_states}) (corresponding to variables $i_k(t)$ and $s_k(t)=1-i_k(t)$) with initial conditions $\psi_k(0)=i_{0k}$. If $i_{0k}\in(0,1]$ then $\psi_k(t),\eta_k(t)$ lie in $[0,1]$ at all times $t\in[0,\infty),~k\in\mathbb K$.\footnote{\emph{Proof:} At any interior point $\{(\psi_k(t),~k\in\mathbb K):0<\psi_k(t)<1,~\forall k\}$, $\frac{d}{dt} {\psi}_k(t)$ is positive $\forall k$ (from Eq. (\ref{eq:opt_prob_states})), hence $\psi_k(t)$ is increasing. However, at any boundary segment $\{(\psi_k(t),~k\in\mathbb K):0<\psi_k(t)\leq 1,~k\in\mathbb K-\{j\} \text{ and } \psi_j(t)=1\}$, for an arbitrary $j\in\mathbb K$,~~$\frac{d}{dt} {\psi}_j(t)=0$ (from Eq. (\ref{eq:opt_prob_states}), as $\eta_j(t)=0$). Once the value of any state variable reaches 1, it stays there. Hence, the solutions $\psi_k(t),~k\in\mathbb K$ always lie in $[0,1]$.}
\end{lemma}

\section{Existence of a Solution}
\label{sec:existence}

It is important to prove the existence of a solution in optimal control problems before attempting to solve them. Even simple looking problems sometimes do not have a solution (examples can be found in \cite[Ch. 3]{fleming1975deterministic}). Existence of a solution to Problem (\ref{eq:opt_prob}) is proved in the following:

\begin{thm}
\label{thm:soln_exist}
There exists a solution $\boldsymbol u^*\in U^{|\mathbb K|}$, $\boldsymbol{i_0^*}$ and the corresponding solution $\boldsymbol i^*(t)$ to the initial value problem (\ref{eq:opt_prob_states}), (\ref{eq:opt_prob_init_cond}) so that $(\boldsymbol u^*, \boldsymbol i_0^*) \in \underset{(\boldsymbol u, \boldsymbol i_0)} {\arg \max}~J(\boldsymbol u, \boldsymbol i_0)$ in the optimal control problem (\ref{eq:opt_prob}).
\end{thm}
\begin{proof}
In Appendix \ref{app:proof_existence} of the supplementary material.
\end{proof}

\section{Analysis and Solution}

In Sec. \ref{sec:soln_by_pontryagin}, we first discuss the solution technique for Problem (\ref{eq:opt_prob}) where seed $\boldsymbol{i_0}$ is a given quantity and not an optimization variable. This will be used in Sec. \ref{sec:soln_joint} to solve the joint problem.

\subsection{Solution by Pontryagin's Maximum Principle (Given Seed $\boldsymbol{i_0}$)}
\label{sec:soln_by_pontryagin}

The optimal solution to Problem (\ref{eq:opt_prob}) (where seed $\boldsymbol{i_0}$ is given) satisfies the conditions stated by Pontryagin's Maximum Principle. Let the adjoint variables be denoted by $\lambda_k(t)$, with the vector $\boldsymbol \lambda(t)=\{\lambda_k(t),~k\in\mathbb K\}$, collecting all adjoint variables. The Pontryagin's Principle applied to our problem leads to the following equations:
\\ \emph{Hamiltonian:}
\begin{align}
& H(\boldsymbol i(t),\boldsymbol \lambda(t),\boldsymbol u(t))=-\sum_{j\in\mathbb K}g_j(u_j(t)) \nonumber\\
& + \sum_{j\in\mathbb K}\lambda_j(t)\left(\beta(t) j s_j(t) \sum_{l\in \mathbb K}(q_li_l(t)) + \gamma(t)u_j(t)s_j(t) \right). \label{eq:hamiltonian}
\end{align}

If $\boldsymbol i^*(t)=\{i_k^*(t),~k\in\mathbb K\}$, $\boldsymbol \lambda^*(t)=\{\lambda_k^*(t),~k\in\mathbb K\}$ and $\boldsymbol u^*(t)=\{u_k^*(t),~k\in\mathbb K\}$ denote the values of the variables at the optimum, then they satisfy the following conditions:
\\ \emph{State equations:} Eq. (\ref{eq:opt_prob_states}) with $i_k(t),s_k(t),u_k(t)$ replaced by $i_k^*(t),s_k^*(t),u_k^*(t)$ respectively.
\\ \emph{Adjoint Equations:} For all $k\in\mathbb K$,
\begin{align}
\frac{d}{dt} \lambda_k^*(t) =& -\frac{\partial}{\partial i_k(t)}H(\boldsymbol i^*(t),\boldsymbol \lambda^*(t),\boldsymbol u^*(t)) \nonumber \\
=& \beta(t) k \lambda_k^*(t)\sum_{l\in\mathbb K}(q_li_l^*(t)) \nonumber \\
& -\beta(t) q_k \sum_{j\in \mathbb K}(\lambda_j^*(t)js_j^*(t)) + \gamma(t)u_k^*(t)\lambda_k^*(t). \label{eq:costate_diff_eq}
\end{align}
\emph{Hamiltonian Maximizing Conditions:} For all $k\in\mathbb K$, the controls satisfy,
\begin{align}
& u_k^*(t) = \text{argmax}_{u_k(t)} H(\boldsymbol i^*(t),\boldsymbol \lambda^*(t),\boldsymbol u(t)) \nonumber \\
\Rightarrow & \frac{\partial}{\partial u_k(t)}H(\boldsymbol i^*(t),\boldsymbol \lambda^*(t),\boldsymbol u^*(t)) \nonumber \\
& = -g_k'(u_k^*(t)) + \gamma(t)\lambda_k^*(t)s_k^*(t)=0, \nonumber \\
\Rightarrow & g_k'(u_k^*(t)) = \gamma(t)\lambda_k^*(t)s_k^*(t), \label{eq:hamiltonian_max_cond_1} \\
\Rightarrow & u_k^*(t) = g_k'^{-1}(\gamma(t)\lambda_k^*(t)s_k^*(t)). \label{eq:hamiltonian_max_cond_2}
\end{align}
\emph{Transversality Conditions:} 
\begin{align}
\lambda_k^*(T)=p_k,~k\in\mathbb K. \label{eq:transversaility_cond}
\end{align}

\subsubsection{Numerical Solution Using Forward-Backward Sweep Algorithm}
\label{sec:numerical_soln_fw_bk_sweep}

Although some structural results may be obtained for the solution to Problem (\ref{eq:opt_prob}) (Sec. \ref{sec:structural_results}), it is unlikely that an analytical solution to the equations in Sec. \ref{sec:soln_by_pontryagin}, and hence, an analytical solution to the control signals, can be obtained. Thus, the equations in Sec. \ref{sec:soln_by_pontryagin} have to be solved numerically to obtain the solution.

Notice that we have a huge optimality system with upto $|\mathbb K|\approx 100$ state equations (with initial conditions), $|\mathbb K|\approx 100$ adjoint equations (with terminal conditions) and $|\mathbb K|\approx 100$ control signals, leading to a boundary value problem with $2|\mathbb K|$ differential equations. However, the above optimality system can be efficiently solved using the forward-backward sweep technique (see for example \cite{asano2008optimal}), using only (numerical) initial value problem solvers. We sketch the technique in Algorithm \ref{alg:fw_back_sweep}. We chose $N_{sweep}=30$, which was sufficient to result in convergence for various sets of parameters used in this work.

An alternate method to solve optimal control Problem (\ref{eq:opt_prob}) is to directly discretize the differential equations into $D$ time points and find the values of the controls at those time instants using an optimization routine \cite{becerra2004solving}. Such a method will not use Pontryagin's Principle. In our experience, such a method becomes extremely slow when large number of controls need to be computed. An additional issue is the large memory requirement for the computation, which may not be available in normal desktop computers.

It is not possible to provide analytical expressions for the gradient of the objective function with respect to the control variables at $D$ discretization points (which are the variables to be optimized); hence any optimization routine approximates the gradient at each optimization-iteration numerically. $D$ cannot be made too low because the accuracy and stability of the solution to the differential equations in System (\ref{eq:opt_prob}) will be compromised. Even for $D=50$ point discretization, the number of variables to be computed is $|\mathbb K|D\approx 5000$. Thus, an optimization routine will need to evaluate the objective function at least $|\mathbb K|D$ times just for estimating the gradient (\emph{e.g.} perturbing one variable at a time). Each objective function evaluation amounts to solving $|\mathbb K|$ differential equations. Also, there are many optimization-iterations before convergence to the solution occurs (typically greater than $N_{sweep}$).

On the other hand, in Algorithm \ref{alg:fw_back_sweep}, one needs to evaluate $2|\mathbb K|$ differential equations only $N_{sweep}=30$ times to obtain the solution. In practice this leads to substantial reduction in computational complexity if a large number of controls have to be computed. Thus, the use of Pontryagin's Principle and forward-backward sweep leads to more efficient computation than a direct discretization method for the problems of the scale discussed in this paper.

\begin{algorithm}
\small
\caption{Forward-backward sweep algorithm for Problem (\ref{eq:opt_prob}).}
\label{alg:fw_back_sweep}
\begin{algorithmic}[1]
	\REQUIRE $N_{sweep}$, $T$; $\beta(t),\gamma(t)~\forall t\in[0,T]$; $i_{0k},~p_k,~q_k,~k\in\mathbb K$.
	\ENSURE The optimal control signals $u_k^*(t),~k\in\mathbb K$.
	\STATE Initialize: $u_k^*(t) \leftarrow 0,~\forall t\in[0,T],~\forall k\in\mathbb K$.
	\FOR{$j=1$ \TO $N_{sweep}$}
		\STATE Calculate $i_k^*,~\forall k\in\mathbb K$ using state equations (\ref{eq:opt_prob_states}) with initial conditions $i_k^*(0)=i_{0k},~\forall k\in\mathbb K$. \COMMENT{Forward sweep.}
		\STATE Calculate $\lambda_k^*,~k\in\mathbb K$ using adjoint equations (\ref{eq:costate_diff_eq}) with terminal conditions $\lambda_k^*(T)=p_k,~k\in\mathbb K$ (transversaility conditions). \COMMENT{Backward sweep.}
		\STATE Calculate $u_k^*,~k\in\mathbb K$ using (\ref{eq:hamiltonian_max_cond_2}).
	\ENDFOR
\end{algorithmic}
\normalsize
\end{algorithm}

The following provides a sufficient condition for the convergence of Algorithm \ref{alg:fw_back_sweep}.
\begin{thm}
\label{thm:conv_fw_bk_sweep}
For $g_k(u_k(t))=c_ku^2_k(t)$, the forward-backward sweep algorithm converges when
\begin{align*}
& \frac{\gamma^2_{M}\Lambda}{2c_{m}} \times \exp\{(\beta_{M}K_{max}+\gamma_{M}u_{M})T\} \times \\
& \Big[ \frac{\exp\{ \beta_{M}(\Sigma k) q_{M} T \} - \exp\{ \beta_{M} K_{max} T \} }{\beta_{M}(\Sigma k) q_{M} - \beta_{M} K_{max}} \Big] < 1,
\end{align*}
where, $\beta_{M}=\max_t\{\beta(t)\},~\gamma_{M}=\max_t \{\gamma(t)\},~c_{m}=\min_k \{c_k\},~u_{M}=\max_{k,t} \{u_k(t)\}, ~\Lambda=\max_{k,t}\{\lambda_k(t)\}$.
\end{thm}
\begin{proof}
Note that the convergence is aided by small values of $\beta(t), \gamma(t)$, and $T$; and large costs of application of controls. We use the techniques in \cite{mcasey2012convergence}. Detailed analysis is in Appendix \ref{app:proof_conv_fw_bk_sweep} of the supplementary material accompanying this paper.
\end{proof}

\subsection{Solution to the Joint Problem (\ref{eq:opt_prob})}
\label{sec:soln_joint}

We use the solution to the fixed seed problem in Sec. \ref{sec:soln_by_pontryagin} and Algorithm \ref{alg:fw_back_sweep} to solve the joint problem. In Problem (\ref{eq:opt_prob}), the solution to the control signals $\boldsymbol{u}$ are functions of $\boldsymbol i_0=\{ i_{0k},~k\in\mathbb K \}$ and are not independent. The joint optimization Problem (\ref{eq:opt_prob}) is equivalent to:
\begin{subequations}
\label{eq:opt_prob_joint_mod}
\begin{align}
\underset{ \begin{smallmatrix} \big\{ \boldsymbol i_0:~0\leq i_{0k}\leq 1, \\ \sum\limits_{k\in\mathbb K} p_ki_{0k}=B_{i_0} \big\} \end{smallmatrix} }{\text{~~maximize~~}} ~J & = \sum_{k\in \mathbb K} p_k i_k(T)-\int_0^T\sum_{k\in \mathbb K}g_k(u_k(t))dt, \label{eq:cost_funtion_joint}\\
\text{subject to:~~} & \text{(\ref{eq:opt_prob_states}) and (\ref{eq:opt_prob_init_cond}); (\ref{eq:costate_diff_eq}) and (\ref{eq:transversaility_cond}); (\ref{eq:hamiltonian_max_cond_2})}. \nonumber \end{align}
\end{subequations}
The constraints in Problem (\ref{eq:opt_prob_joint_mod}) ensure that Pontryagin's Principle is satisfied for the control functions computed by (\ref{eq:hamiltonian_max_cond_2}) for any value of $\boldsymbol{i_0}$ (thus the computed controls are optimal).

Problem (\ref{eq:opt_prob_joint_mod}) (which in turn solves Problem (\ref{eq:opt_prob})) can be solved numerically by combination of an optimization solver and Algorithm \ref{alg:fw_back_sweep}. The optimization routine adjusts the values of optimization variables $\boldsymbol{i_0}$ in the outer loop, and the reward function is computed by Algorithm \ref{alg:fw_back_sweep} (for the given value of $\boldsymbol{i_0}$). It is not possible to compute the gradient of the reward function (\ref{eq:cost_funtion_joint}) with respect to $\boldsymbol{i_0}$ analytically, so the optimization routine should be capable of numerically estimating the gradient values.

In principle, it may be possible to avoid the use of Algorithm \ref{alg:fw_back_sweep} and Pontryagin's Principle in solving Problem (\ref{eq:opt_prob}) by using the method in \cite{becerra2004solving} and augmenting $\boldsymbol{i_0}$ as additional optimization variables. But as discussed in Sec. \ref{sec:numerical_soln_fw_bk_sweep}, such a method will be too slow and will require a large amount of memory for computation.

\subsection{Structural and Uniqueness Results for Controls}
\label{sec:structural_results}

In this section, we prove some basic structural results for the solution to Problem (\ref{eq:opt_prob}) when seeds are given (Theorem \ref{thm:controls_structure}) and provide a sufficient condition for uniqueness of controls (Theorem \ref{thm:uniqueness_of_state_adjoint_opt_control}). We first provide Lemma \ref{thm:adjoint_variables_positive} which is needed in subsequent results:
\begin{lemma}
\label{thm:adjoint_variables_positive}
The adjoint variables at the optimum, $\boldsymbol \lambda^*(t)=\{\lambda_k^*(t),~k\in\mathbb K\}$ satisfy $\lambda_k^*(t)\geq 0,~\forall k\in\mathbb K,~\forall t\in[0,T]$.
\end{lemma}
\begin{proof}
In Appendix \ref{app:proofs} of the supplementary material.
\end{proof}

\begin{thm}
\label{thm:controls_structure}
(i) The solutions $\{u_k^*,~k\in\mathbb K\}$ to the optimal control problem (\ref{eq:opt_prob}) are non-increasing functions of time $\forall k\in\mathbb K$ and $t\in[0,T]$ for $\frac{d}{dt} \gamma(t)\leq 0$.

(ii) For $g_k(u_k(t))=c_ku_k^2(t)$, $c_k>0$, the solutions $\{u_k^*,~k\in\mathbb K\}$ to the optimal control problem (\ref{eq:opt_prob}) are convex functions of time $\forall k\in\mathbb K$ and $t\in[0,T]$, for constant and non-increasing spreading rate profiles and effectiveness of controls, \emph{i.e.}, $\frac{d}{dt}{\beta}(t), \frac{d}{dt} \gamma(t) \leq 0$, and convex $\gamma(t)$, \emph{i.e.}, $\frac{d^2}{dt^2}{\gamma}(t)\geq 0$  (this includes $\frac{d}{dt}{\beta}(t),\frac{d}{dt}\gamma(t)=0,~\forall t$).
\end{thm}
\begin{proof}
In Appendix \ref{app:proofs} of the supplementary material.
\end{proof}
\textbf{Remark:} Lemma \ref{thm:adjoint_variables_positive} and Theorem \ref{thm:controls_structure}(i) are valid for a spreading rate profile of arbitrary shape, \emph{i.e.}, for any $\beta(t)\geq 0,~t\in[0,T]$. Further, Lemma \ref{thm:adjoint_variables_positive} and Theorem \ref{thm:controls_structure} are valid for any degree distribution $p_k$, in other words, for any configuration model network.

When seeds are given, we state a sufficient condition for a solution of Problem (\ref{eq:opt_prob}) to be unique.
\begin{thm}
\label{thm:uniqueness_of_state_adjoint_opt_control}
When seeds are given, for $g_k(u_k(t))=c_ku_k^2(t)$, the solution to the state equations (\ref{eq:opt_prob_states}) and adjoint equations (\ref{eq:costate_diff_eq}), and hence the optimal controls for Problem (\ref{eq:opt_prob}) are unique when \begin{equation}
\label{eq:cond_for_uniqueness}
d_1||\beta(t)||_{L_1}+d_2||\gamma^2(t)||_{L_1}<1.
\end{equation}
Here, $d_1=\max\{ (\sum_{k\in\mathbb K}k)\Lambda q_{M} + K_{max}\Lambda, 2K_{max} \},~d_2=(\Lambda/c_{m})\max\{ 1,\Lambda/2 \},~\Lambda=\max_{k,t}\{\lambda_k(t)\},~q_{M}=\max_k\{ q_k \},~c_{m}=\min_k\{ c_k \}$ and $||.||_{L_1}$ is $L_1$ norm of a function.\footnote{The result can be generalized for any $g_k()$ by assuming Lipschitz continuity of $g_k'^{-1}(.)$ and assuming a Lipschitz constant.}
\end{thm}
\begin{proof}
There are multiple approaches to show uniqueness of solutions. We have used the techniques in \cite{ma2002existence} instead of those in, for example \cite{kandhway2014run}, because the former leads to more insightful conclusions. The proof is in Appendix \ref{app:proof_uniqueness} of the supplementary material accompanying this paper. 
\end{proof}
\textbf{Remark:} Systems with small $L_1$ norms for $\beta(t)$, $\gamma(t)$ (which depend both on function values and $T$), small values of $\Sigma k, q_{M}, K_{max}$; and large values of cost of application of controls $c_k$, have unique solutions.

\section{A Control Problem With a Fixed Budget Constraint}

\subsection{Problem Formulation and Solution by Pontryagin's Principle}

In many practical scenarios, campaign resources are fixed, \emph{e.g.} political campaigns. For such cases, we modify Problem (\ref{eq:opt_prob}) (where seed vector is a given quantity) to include an explicit budget constraint (Eq. \ref{eq:opt_prob_budget_constraint}) in the following:
\begin{subequations}
\label{eq:opt_prob_budget}
\begin{align}
\underset{\boldsymbol u \in U^{|\mathbb K|}}{\text{~~maximize~~}} ~J & = \sum_{k\in \mathbb K} p_k i_k(T) \label{eq:cost_funtion_budget}\\
\text{subject to:~~} & \text{(\ref{eq:opt_prob_states}) and (\ref{eq:opt_prob_init_cond}),} \nonumber \\
& \int_0^T\sum_{k\in \mathbb K}g_k(u_k(t))dt - B \leq 0. \label{eq:opt_prob_budget_constraint}
\end{align}
\end{subequations}
We cannot use Pontryagin's principle in Problem (\ref{eq:opt_prob_budget}) due to the integral constraint (\ref{eq:opt_prob_budget_constraint}). However, the budget constraint (\ref{eq:opt_prob_budget_constraint}) can be handled by the standard optimization technique of relaxing the inequality constraint into the objective function. Problem (\ref{eq:opt_prob_budget}) can then be re-written as:
\begin{subequations}
\label{eq:opt_prob_budget_relaxed}
\begin{align}
\underset{\boldsymbol u \in U^{|\mathbb K|}} {\text{~~max~~}} ~J & = \sum_{k\in \mathbb K} p_k i_k(T)-\mu \left( \int_0^T\sum_{k\in \mathbb K}g_k(u_k(t))dt - B \right) \label{eq:cost_funtion_budget_relaxed}\\
\text{subject to:~~} & \text{(\ref{eq:opt_prob_states}) and (\ref{eq:opt_prob_init_cond}).} \nonumber
\end{align}
\end{subequations}

Problem (\ref{eq:opt_prob_budget_relaxed}) solves Problem (\ref{eq:opt_prob_budget}) for the value of the multiplier $\mu^*$ for which constraint (\ref{eq:opt_prob_budget_constraint}) is satisfied. Also, for a given value of $\mu$, $\mu B$ is just a constant and can be eliminated from the objective function (\ref{eq:cost_funtion_budget_relaxed}). Pontryagin's Principle applied to Problem (\ref{eq:opt_prob_budget_relaxed}) produces the same equations as Sec. \ref{sec:soln_by_pontryagin} with two differences: the Hamiltonian in Eq. (\ref{eq:hamiltonian}) is replaced by:
\begin{align}
& H(\boldsymbol i(t),\boldsymbol \lambda(t),\boldsymbol u(t))=-\sum_{j\in\mathbb K}\mu g_j(u_j(t)) \nonumber\\
& + \sum_{j\in\mathbb K}\lambda_j(t)\left(\beta(t) j s_j(t) \sum_{l\in \mathbb K}(q_li_l(t)) + \gamma(t)u_j(t)s_j(t) \right). \nonumber
\end{align}
Further, the equations in Hamiltonian maximization condition, (\ref{eq:hamiltonian_max_cond_1}) and (\ref{eq:hamiltonian_max_cond_2}), are replaced by:
\begin{align}
& g_k'(u_k^*(t))=\gamma(t)\lambda_k^*(t)s_k^*(t)/\mu^*, ~k\in\mathbb K, \nonumber \\
\Rightarrow & u_k^*(t)=g_k'^{-1}(\gamma(t)\lambda_k^*(t)s_k^*(t)/\mu^*), ~k\in\mathbb K. \label{eq:hamiltonian_max_cond_budget}
\end{align}
The state and adjoint equations, and the Transversaility conditions are the same as in Sec. \ref{sec:soln_by_pontryagin}.

\subsection{Structural Results}

\begin{lemma}
\label{thm:multiplier_positive}
At the optimum, the multiplier of the relaxed problem (\ref{eq:opt_prob_budget_relaxed}) which solves Problem (\ref{eq:opt_prob_budget}), $\mu^*\geq 0$ and constraint (\ref{eq:opt_prob_budget_constraint}) is satisfied with equality.\footnote{\emph{Proof:} From the standard optimization theory we know that relaxing the inequality constraint (\ref{eq:opt_prob_budget_constraint}) leads to the value of multiplier $\mu^*$ at the optimum being non-negative, $\mu^*\geq 0$. At optimum, the whole budget $B$ is used; if this is not the case, we can increase one or more of the $u_k$'s in (\ref{eq:opt_prob_budget_constraint}), so that the constraint is still satisfied. Doing so will increase $i_k(T)$'s (we can conclude this from state equations (\ref{eq:opt_prob_states})) and hence the value of the objective function (\ref{eq:cost_funtion_budget}). Hence, the budget $B$ is not underutilized at the optimum.
}
\end{lemma}

Due to Lemma \ref{thm:multiplier_positive}, Lemma \ref{thm:adjoint_variables_positive} and Theorem \ref{thm:controls_structure} are valid even for the solution to the optimal control problem (\ref{eq:opt_prob_budget}) (replace $g_k(u_k^*(t))$, $g_k'(u_k^*(t))$ and $g_k''(u_k^*(t))$ by $\mu^* g_k(u_k^*(t))$, $\mu^* g_k'(u_k^*(t))$ and $\mu^* g_k''(u_k^*(t))$ respectively in the proofs).

\subsection{Numerical Solution}

To solve Problem (\ref{eq:opt_prob_budget_relaxed}) (for which $\mu^*$ satisfies constraint (\ref{eq:opt_prob_budget_constraint})), which in turn solves Problem (\ref{eq:opt_prob_budget}), we modify the standard forward-backward sweep algorithm. The standard algorithm cannot handle the budget constraint. Our approach is to adjust the value of the multiplier $\mu^*$, using the bisection algorithm, till constraint (\ref{eq:opt_prob_budget_constraint}) is satisfied with equality (due to Lemma \ref{thm:multiplier_positive}). We sketch the procedure in Algorithm \ref{alg:fw_back_sweep_modified}. The values $\mu_{low}=10^{-3},~\mu_{high}=100$ and $N_{sweep}=30$ were used for all computations in the results section.

\begin{algorithm}
\small
\caption{Modified forward-backward sweep algorithm for Problem (\ref{eq:opt_prob_budget_relaxed}) (for which $\mu^*$ satisfies constraint (\ref{eq:opt_prob_budget_constraint}) with equality).}
\label{alg:fw_back_sweep_modified}
\begin{algorithmic}[1]
	\REQUIRE $\mu_{low},~\mu_{high},~B,~N_{sweep}$; $i_{0k},~p_k,~q_k,~\forall k\in\mathbb K$; $T$ and $\beta(t),\gamma(t)~\forall t\in[0,T]$.
	\ENSURE The optimal control signals $u_k^*(t)$ and value of the multiplier $\mu^*$ for which (\ref{eq:opt_prob_budget_constraint}) is satisfied with equality.
	\STATE $B_{th}\leftarrow \min\{10^{-3}\times B,10^{-6}\}$
	\REPEAT
		\STATE $\mu^* \leftarrow (\mu_{low} + \mu_{high})/2$.
		\STATE Calculate $i_k^*,\lambda_k^*$ and $u_k^*$ using Algorithm \ref{alg:fw_back_sweep}, however, replacing Eq. (\ref{eq:hamiltonian_max_cond_2}) by Eq. (\ref{eq:hamiltonian_max_cond_budget}).
		\STATE $r_{\mu^*} \leftarrow \int_0^T\sum_{k\in \mathbb K}g_k(u_k^*(t))dt$ \COMMENT{Resource used.}
		\IF{$r_{\mu^*}>B$}
			\STATE $\mu_{low} \leftarrow \mu^*$
		\ENDIF
		\IF{$r_{\mu^*}<B$}
			\STATE $\mu_{high} \leftarrow \mu^*$
		\ENDIF
	\UNTIL $\big|r_{\mu^*} - B\big|<B_{th}$.
\end{algorithmic}
\normalsize
\end{algorithm}

\section{Synthetic Networks and Model Parameters}

\emph{Networks:} We present results by using degree distributions from two synthetic networks (in Sec. \ref{sec:results}) and a real network (in Sec. \ref{sec:results_real_nw}) in this study. The first synthetic network is an Erd\H os-R\'enyi network which has degree distribution $p_k=e^{-\lambda}\lambda^k/k!,~k\in\mathbb K$. We choose the minimum and maximum degrees in $\mathbb K^{ER}$ as $K_{min}^{ER}=13$ and $K_{max}^{ER}=54$, so that the truncated degrees have very less cumulative probability. The factor $\lambda=33.45$ is the same as the mean degree $\bar k^{ER}$ for this network.

The second synthetic network is scale-free which has a power law degree distribution, $p_k=\omega k^{-\alpha},~k\in\mathbb K$. Here $\omega$ normalizes the distribution to $1$ and $\alpha$ is the power law exponent. Most real networks---including the internet, the world wide web, and more importantly the social networks---have power law exponent lying between $2$ and $3$ \cite{newman2009networks}. We have chosen $\alpha=2$ for our scale-free network. The minimum and maximum degrees in the power law distribution are chosen as $K_{min}^{PL2}=14$, $K_{max}^{PL2}=120$. Degrees are adjusted so that both PL2 and ER network above have almost the same mean degrees. The mean degree for PL2, $\bar k^{PL2}=33.29$. For networks of the same size, having equal mean degree means none of them has any statistical advantage in information dissemination due to more number of links. As suggested by all the problem formulations in this paper, we only need degree distributions of nodes in the networks for presenting numerical results. Maximum degree in the scale free network was chosen following the Dunbar's number which states that most humans only maintain $100$ to $230$ stable relationships at a time.

\emph{Default Model Parameters:} In the SI process, the deadline $T$ and the spreading rate $\beta(t)$ determine the extent of information spread in the uncontrolled system. Fixing one and increasing the other has the same qualitative effect. We choose to fix the deadline at $T=1$ time units and vary the spreading rate whenever necessary. When seed is not an optimization variable the initial fraction of infected nodes in each degree class, $i_{0k}=0.01,~\forall k\in\mathbb K$. In other cases seed budget, $B_{i_0}=0.01$. For the plots studying the effect of system parameters, we have used $\beta(t)=\beta=0.07$. Such a choice (with $T=1$) leads to small to moderate information spread (quantified by $i(T)=\sum_{k\in\mathbb K} p_ki_k(T)$) in the uncontrolled system in both the networks ($i(T)_{ER}=0.095$ and $i(T)_{PL2}=0.149$). In a practical scenario, such a case will call for campaigning. Throughout this paper $\gamma(t)=10\times\beta(t).$

To demonstrate the results, the instantaneous cost of application of control is chosen to be $g_k(u_k(t))=bu_k^2(t)p_k$ in (\ref{eq:cost_funtion}). The control strength $u_k(t)$ incurs a cost $u_k^2(t)$. Since a degree class with more nodes will consume more resource, so weighting factor $p_k$ is also present. The parameter $b$ captures the relative importance of reward due to information spread (given by $\sum_{k\in\mathbb K} p_ki_k(T)$) and cost of application of control in degree class $k$ (given by $u_k^2(t)p_k$). For demonstrating results we set $b=25$.

For the problem involving the budget constraint (Problem (\ref{eq:opt_prob_budget})), we assume the same cost structure with $g_k(u_k(t))=bu_k^2(t)p_k$ and $b=25$, and the value of budget $B=0.1$ in the budget constraint (\ref{eq:opt_prob_budget_constraint}).

\subsection{Accuracy of Degree Based Compartmental Model for Synthetic Networks}

\begin{figure}[ht!]
\centering
\includegraphics[width=48mm]{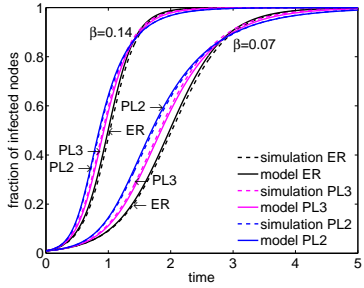}
\caption{\small{The fraction of infected population (measured by $\sum_{k\in\mathbb K} i_k(t)p_k$) vs. time, produced by the model and simulations. Parameter values: $i_{0k}=0.01~\forall k\in\mathbb K,~\beta=0.07$ and $0.14.$ PL3 network has power law exponent $\alpha=3$, $K_{min}^{PL3}=20, K_{max}^{PL3}=120$ which leads to the mean degree of $33.58$.}}
\label{fig:sim_vs_model_SI_config_model}
\end{figure}

The degree based compartmental model for SI epidemic is very accurate for configuration model networks. This is confirmed by the close agreement between the system evolution, measured by the fraction of infected population, generated by the model and simulation for two values of the spreading rates $\beta(t)=\beta$, for the three networks ER, PL3 and PL2, in Fig. \ref{fig:sim_vs_model_SI_config_model}. PL3 network has power law exponent $\alpha=3$, $K_{min}^{PL3}=20, K_{max}^{PL3}=120$ which leads to the mean degree of $33.58$. We have used the uncontrolled SI model. The simulation results are averaged over $20$ runs for all six curves (corresponding to `simulation') in the figure.

For each run, a different configuration model network was drawn from the degree distribution of the network under consideration, and a different set of initial nodes were selected as seeds. Sizes of the networks were $10^4$ nodes in all the cases. To construct configuration model networks of size $N$, we follow the procedure in \cite[Sec. 13.2]{newman2009networks}\footnote{The procedure is as follows: The degree distribution, $p_k$, is sampled $N$ times to get degrees of $N$ nodes. Each node is assumed to have half edges equal to their degree. Two half edges are randomly paired to create an edge until all half edges are exhausted. If last half edge is left unpaired, it is ignored. This procedure may lead to multiple edges between two nodes and self loops. However, the fractions of these edges are low (and approaches $0$ as $N\rightarrow\infty$), and hence not an issue \cite[Sec. 13.2]{newman2009networks}.}. Since the degree distribution for a particular network is fixed and uniform seeding is assumed, the model predicts the same trajectory for all 20 runs; the plot corresponding to `simulation' is the mean of the 20 runs.

The simulation results for degree based compartmental model of SI epidemic on a real network is shown later in Sec. \ref{sec:results_real_nw}.

\section{Results}
\label{sec:results}

We will see in this section that the importance of degree classes in the optimal campaigning strategy depends not only on the network topology but also on the system parameters such as the scarcity/abundance of resources and spreading rate. For scarce resource case, the optimal strategy targets highest degrees in the PL2 network but medium degrees in the ER network. When resource is abundant or spreading rate high, lower degrees (which are disadvantaged in receiving the messages) are also directly targeted. In the joint problem, optimal seeding strategy shows a similar behavior.

\subsection{Heuristic Controls}
\label{sec:heuristic_controls}
For Problem (\ref{eq:opt_prob}), we compare the performance of the optimal control strategy with two other strategies. \emph{The first one} is the best `constant or static control'. It is the control which is constant over time and is applied to all degree classes. The strength of the static control is chosen so that it maximizes the reward function (\ref{eq:cost_funtion}) (when seeds are uniformly selected).

\emph{The second one} is the best `two-stage control'. It is constant in $[0, T/2]$ and $0$ in $(T/2,T]$. Again, the strength of non-zero part is chosen such that (\ref{eq:cost_funtion}) is maximized (under uniform seeding) and the same control is applied to all degree classes. From Theorem \ref{thm:controls_structure}(i) the best time to apply the controls is the initial period of the campaign, which motivates such a control. Both the heuristic controls are computed numerically by an optimization routine.

For positive arguments, the cost function is strictly convex-increasing in nature, more control strength means superlinear resource expenditure. So, there is a tradeoff between the two heuristic strategies: the static applies milder control strength but for the entire duration $[0,T]$, while the two-stage control applies stronger control but only for half the duration (although during important times) for the same amount of resource expenditure. Note that the best two-stage control is a simple \emph{dynamic} (time varying) control.

A minor difference exists in Problem (\ref{eq:opt_prob_budget}). Due to the fixed budget constraint, the static and two-stage strategies can be uniquely calculated. The same are used in Sec. \ref{sec:result_effect_parameter_budget}.

\subsection{Shapes of Controls and Importance of Degree classes (Uniform Seeding)}
\label{sec:result_shape_control_imp_deg_classes}

\begin{figure*}[ht!]
\centering
\subfloat[ER, $b=25, \beta=0.07$ \label{fig:control_norm_res_b25_ER}]{
\includegraphics[width=46mm]{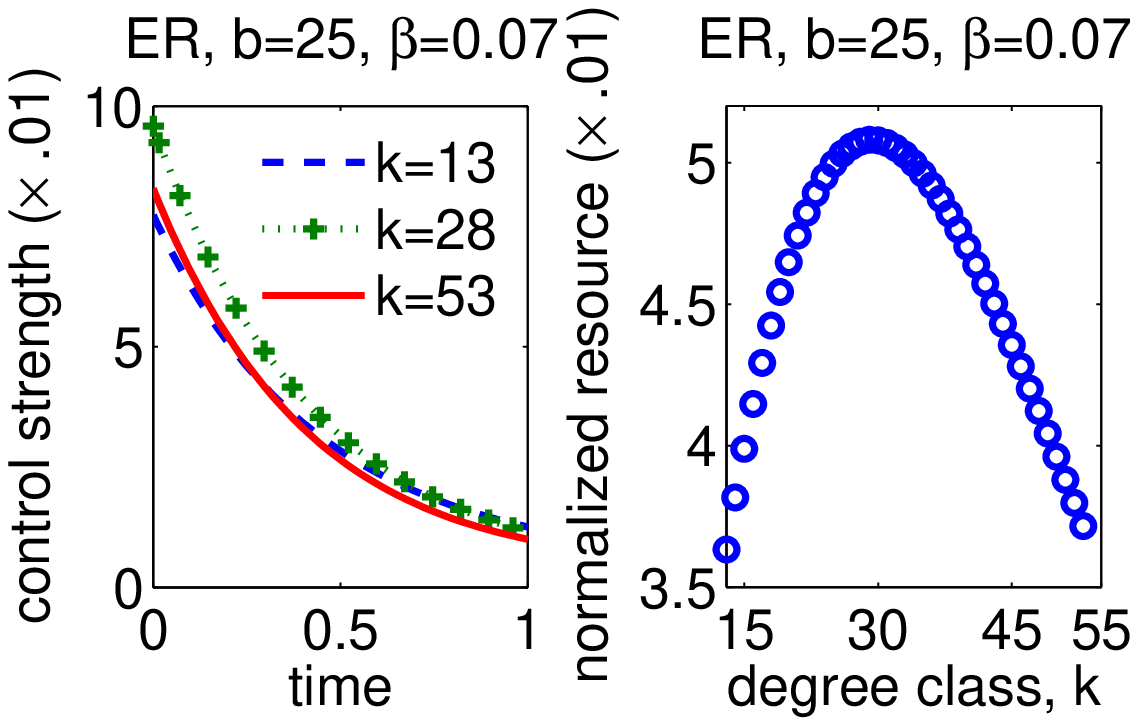} }
\hspace{1em}
\subfloat[ER, $b=0.2, \beta=0.07$ \label{fig:control_norm_res_b_pt2_ER}]{
\includegraphics[width=46mm]{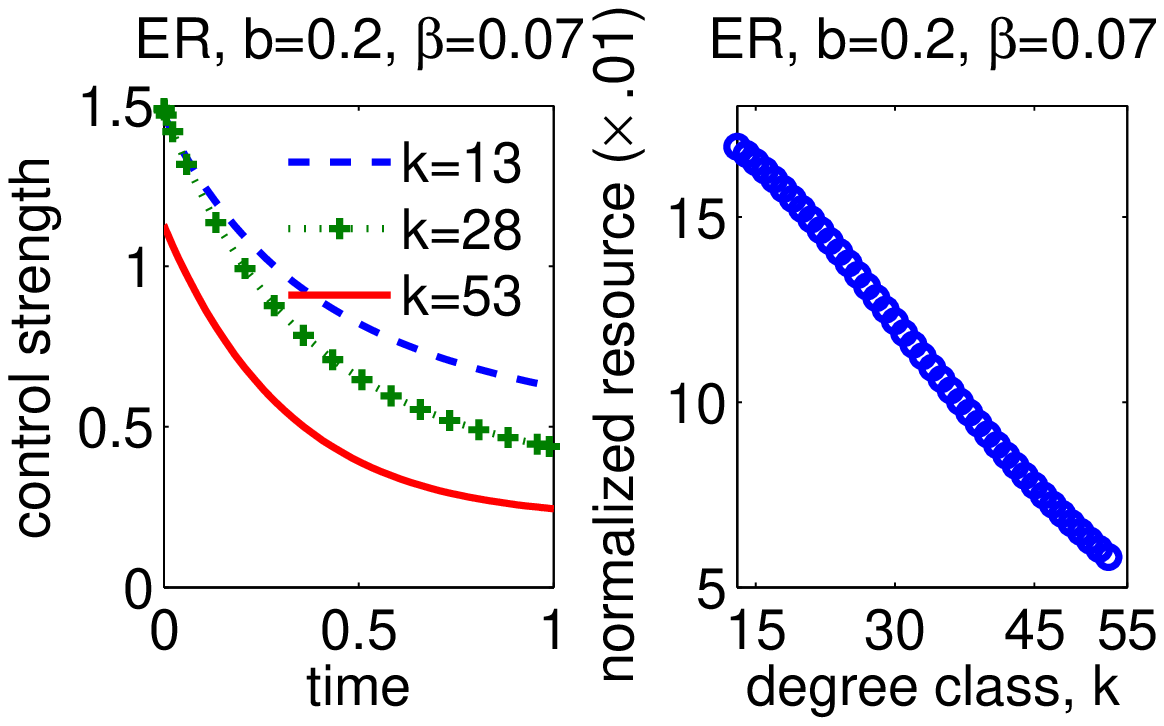} }
\hspace{1em}
\subfloat[ER, $b=25, \beta=0.21$ \label{fig:control_norm_res_b25_beta_pt18_ER}]{
\includegraphics[width=46mm]{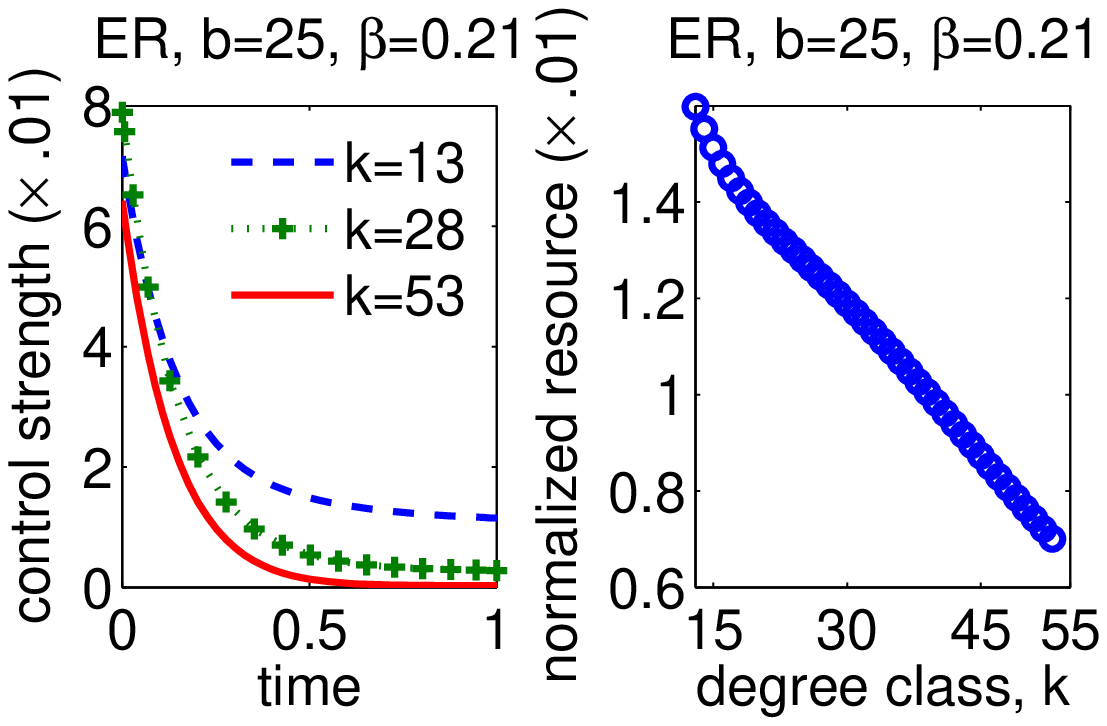} }

\subfloat[PL2, $b=25, \beta=0.07$ \label{fig:control_norm_res_b25_PL2}]{
\includegraphics[width=46mm]{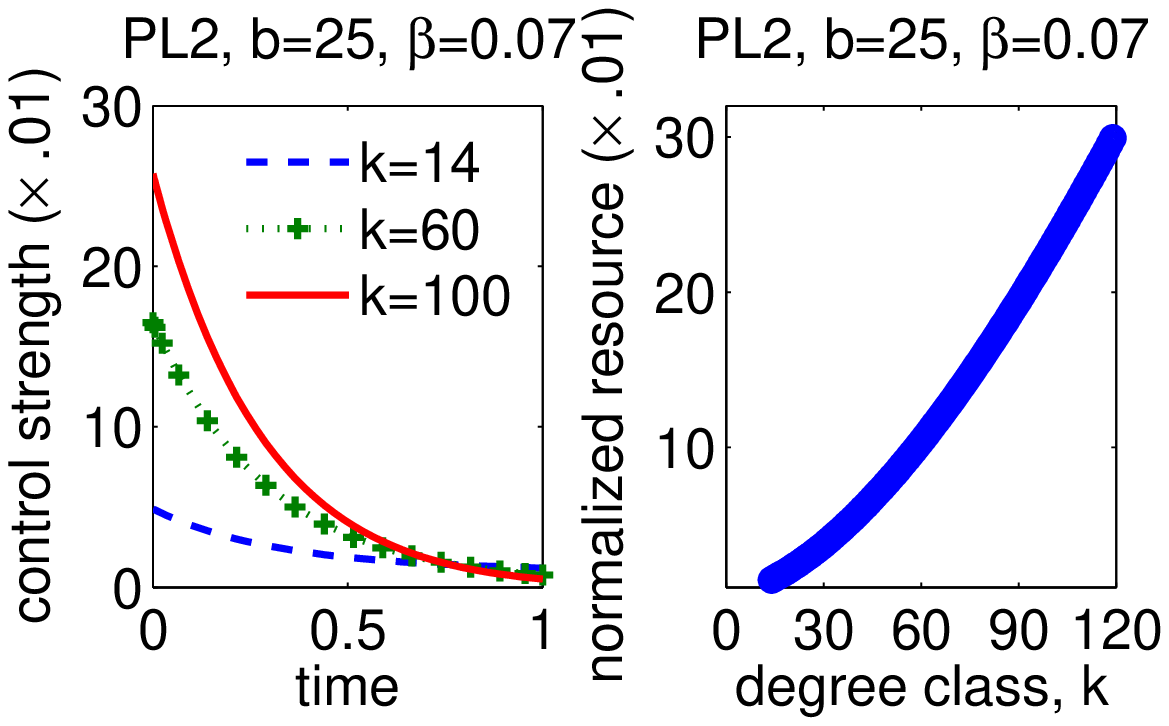} }
\hspace{1em}
\subfloat[PL2, $b=0.2, \beta=0.07$ \label{fig:control_norm_res_b_pt2_PL2}]{
\includegraphics[width=46mm]{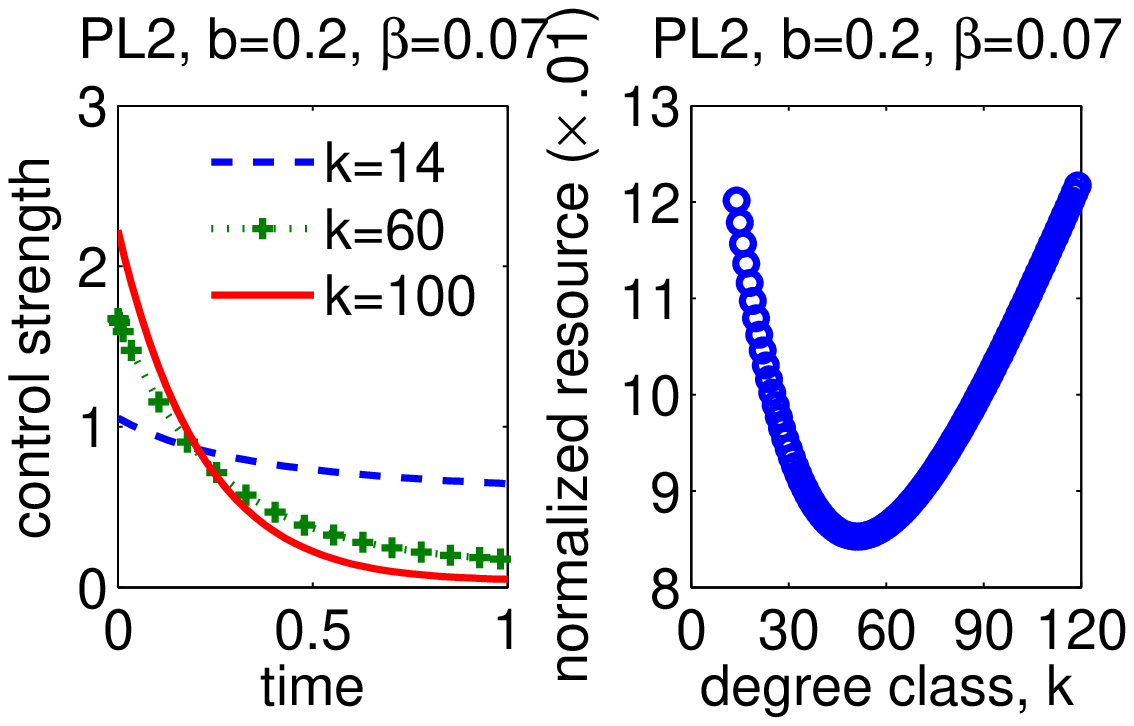} }
\hspace{1em}
\subfloat[PL2, $b=25, \beta=0.21$ \label{fig:control_norm_res_b25_beta_pt21_PL2}]{
\includegraphics[width=46mm]{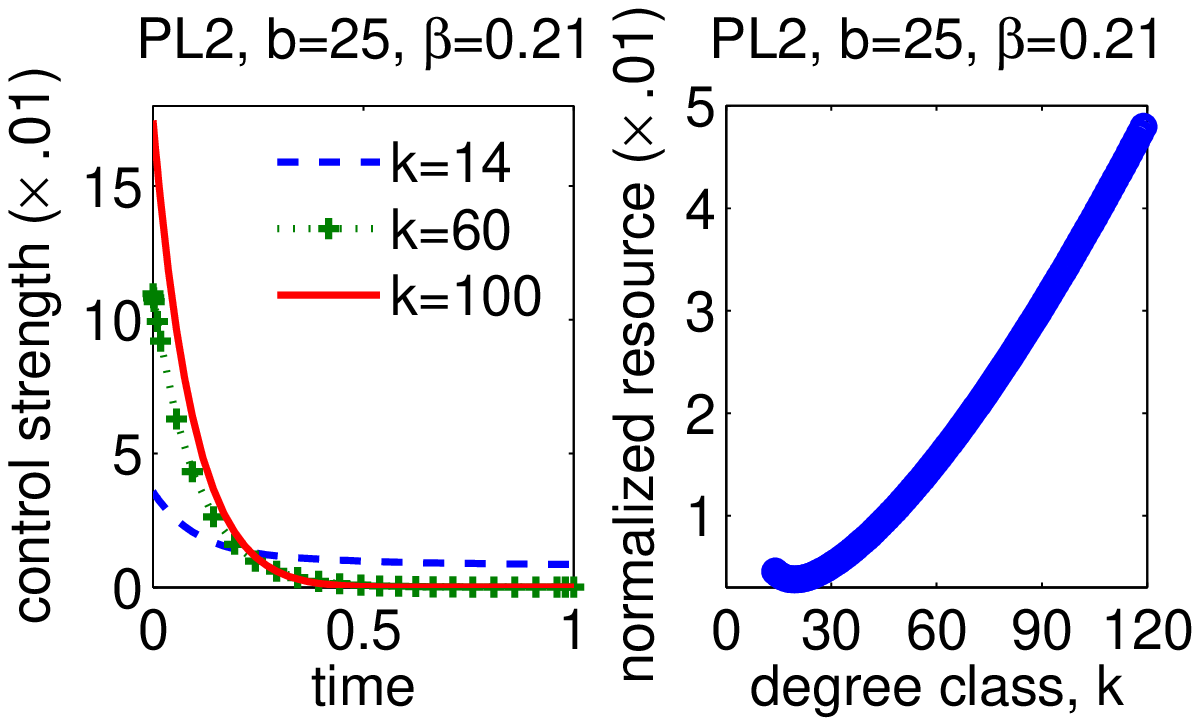} }
\caption{\small{Controls and normalized resource allocation (defined in Eq. (\ref{eq:norm_resource})) for $i_{0k}=i_0=0.01,~\forall k\in\mathbb K$, $\gamma(t)=10\times\beta(t)$.}}
\label{fig:control_norm_res_b25}
\end{figure*}

\begin{figure*}[ht!]
\centering
\subfloat[PL2, increasing $\beta(t)$, $b=25$ \label{fig:control_norm_res_beta1_b25_PL2}]{
\includegraphics[width=46mm]{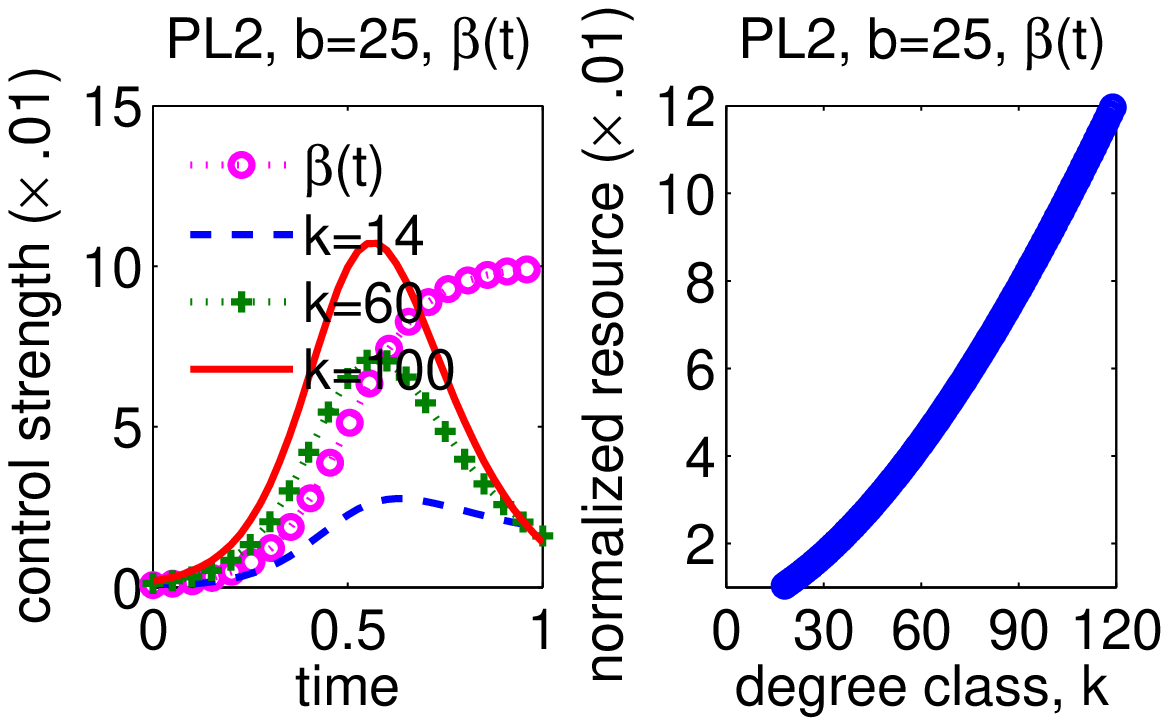} }
\hspace{1em}
\subfloat[PL2, increasing $\beta(t)$, $b=0.2$ \label{fig:control_norm_res_beta1_b_pt2_PL2}]{
\includegraphics[width=46mm]{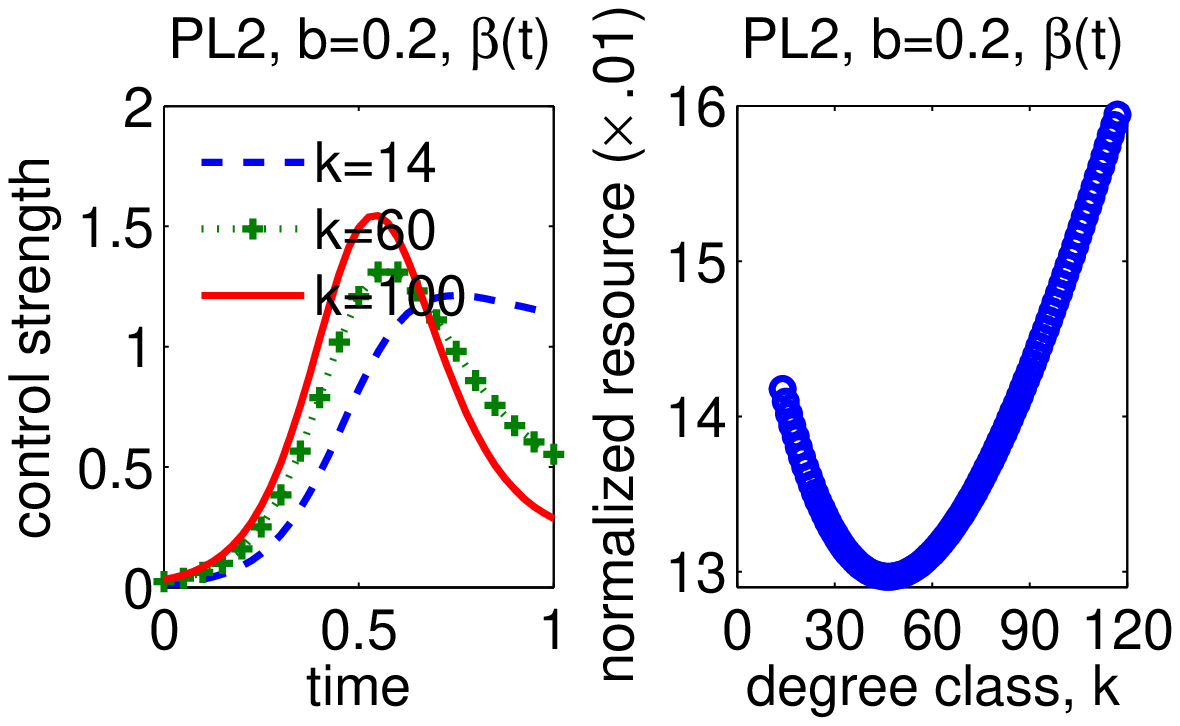} }
\hspace{1em}
\subfloat[PL2, increasing $3\times\beta(t)$, $b=25$ \label{fig:control_norm_res_b25_3beta1_PL2}]{
\includegraphics[width=46mm]{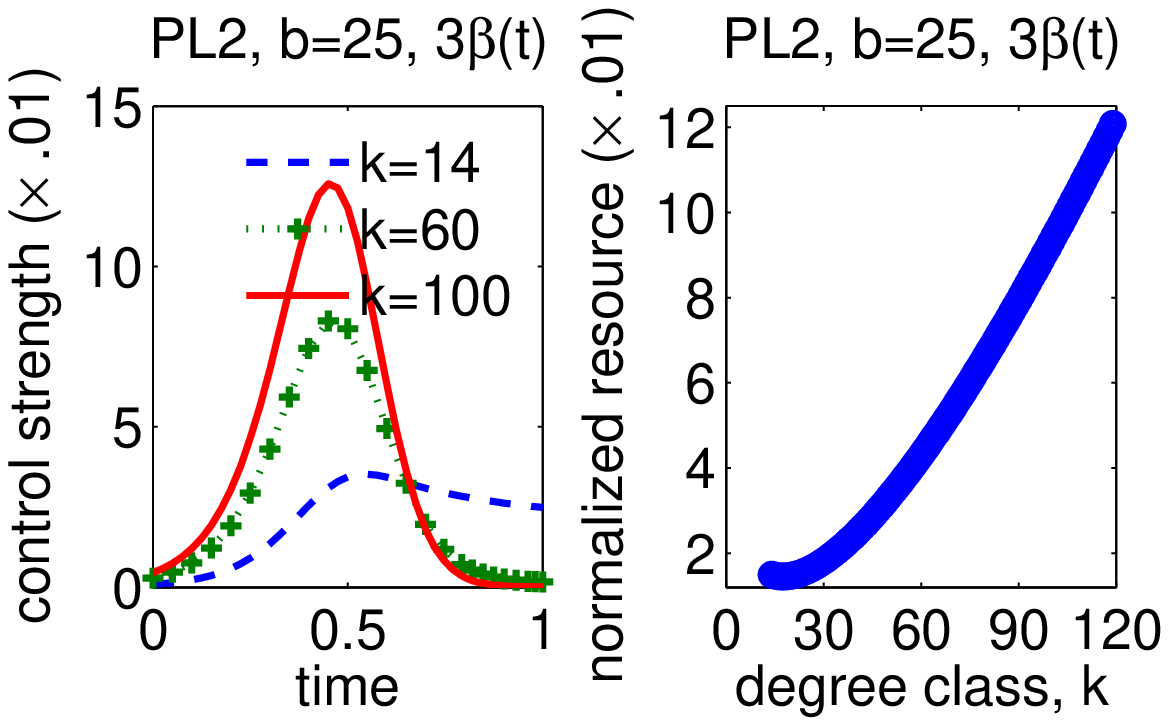} }
\caption{\small{Controls and normalized resource allocation for increasing $\beta(t)$ (as shown above in Fig. \ref{fig:control_norm_res_beta1_b25_PL2}), $i_{0k}=i_0=0.01,~\forall k$, $\gamma(t)=10\times\beta(t)$.}}
\label{fig:control_norm_res_beta1}
\end{figure*}

Figs. \ref{fig:control_norm_res_b25_ER} and \ref{fig:control_norm_res_b25_PL2} (left panels) show the shapes of control signals, $u_k(t)$ for ER and PL2 networks for three representative degree classes for $\beta(t)=\beta=0.07, \gamma(t)=10\times \beta(t)$. They are solutions to Problem (\ref{eq:opt_prob}) when seeds are uniformly selected from the population and are not optimization parameters, \emph{i.e.} $i_{0k}=i_0,~\forall k\in\mathbb K$. The figures show that the controls are non-increasing functions of time (Theorem \ref{thm:controls_structure}(i)). Such a behavior is expected because early infection enhances further information spread by susceptible-infected epidemic contact during rest of the campaign period. Also, for the cost structure, $g_k(u_k(t))=bu_k^2(t)p_k$, in (\ref{eq:cost_funtion}) and constant $\beta(t), \gamma(t)$ the controls are convex functions of time (Theorem \ref{thm:controls_structure}(ii)), which is also confirmed by the figures.

The right panels of Figs. \ref{fig:control_norm_res_b25_ER} and \ref{fig:control_norm_res_b25_PL2} shows the normalized resource allocated to degree class $k$ for the whole campaign period for the ER and PL2 networks considered in this study. Normalized resource is calculated as:
\begin{align}
r^{norm}_k=\frac{1}{p_k}\int_0^Tg_k(u_k(t))dt=b\int_0^Tu_k^2(t)dt. \label{eq:norm_resource}
\end{align}
Note that $r^{norm}_k$ represents per capita resource allocated to each node in the degree class $k$ during campaigning and thus is a proxy for the importance of nodes in the $k$th degree class in information dissemination.

As seen from the normalized resource allocation plots in Figs. \ref{fig:control_norm_res_b25_ER} and \ref{fig:control_norm_res_b25_PL2}, for $b=25$, the \emph{heterogeneous} scale-free PL2 network has different allocation from the \emph{homogeneous} ER network. For the scale-free network, higher degree classes get more per capita resource than lower degree classes; and in ER, medium degree nodes receive most per capita allocation of the campaigning resource. We discuss this in the following.

Direct recruitment balances two things---it targets the degree classes which will lead to further information spread and it targets susceptible nodes which are at a disadvantage in receiving the message from epidemic spreading and directly transfers them to the infected class (to increase the net fraction of infected nodes at the deadline, thus increasing the reward). Scale-free networks have a heavy tail, meaning that, there are sizable numbers/fractions of nodes with high degrees. These nodes have disproportionate advantage in spreading the message due to their large degrees and are often termed as hubs. The optimal strategy allocates them more per capita resource. Targeting them early in the campaign leads to larger diffusion of the message due to more susceptible-infected epidemic contacts during the campaign period due to their large degree.

On the other hand, in the ER network, node degree is concentrated tightly around the mean. Thus, higher degree nodes do not have a significant advantage over other nodes in spreading the message. Due to their larger degree they will any way receive the message so direct recruitment targets medium degree nodes. In the ER network, medium degree nodes are decent spreaders and will indirectly transfer the message to high degree nodes.  In addition, their direct recruitment leads to increased fraction of infected population at the deadline (due to direct transfer of susceptible nodes to the infected class). Low degrees are not targeted because disadvantage due to less spreading is not offset by the advantage due to their direct recruitment, as is the case with medium degrees.

The tradeoff between targeting better spreaders and nodes which are at a disadvantage in receiving the message becomes clearer when the resource becomes cheap (and hence is abundant) which allows us to reach greater fraction of population. Shapes of controls and per capita resource allocated to various degree classes for the case when $b=0.2$ are shown in Figs. \ref{fig:control_norm_res_b_pt2_ER} and \ref{fig:control_norm_res_b_pt2_PL2}. For the heterogeneous PL2 network, low degrees are given more importance than medium degrees. In the homogeneous ER network, low degrees are given most importance.

This behavior is also seen in the ER network (and to a very small extent in the PL2 network) when spreading rate is increased (Figs. \ref{fig:control_norm_res_b25_beta_pt18_ER} and  \ref{fig:control_norm_res_b25_beta_pt21_PL2}). Similar to the above, high spreading rate allows us to reach large fraction of population, so the optimal strategy targets the disadvantaged nodes.

\emph{Time varying $\beta(t),\gamma(t)$:} It is expected that when effectiveness of controls $\gamma(t)$ is a time varying quantity, more resource will be allocated when it is stronger. For a spreading rate profile $\beta(t)$ which varies as an S shaped sigmoid function (shown in the left panel of Fig. \ref{fig:control_norm_res_beta1_b25_PL2}) and $\gamma(t)=10\times\beta(t)$, the controls and normalized resource allocations for $b=25$ and $0.2$ are shown in Fig. \ref{fig:control_norm_res_beta1} for PL2 network. The controls still try to infect nodes early; however, they wait till $\gamma(t)$ becomes strong enough. This leads to more efficient utilization of resources. Note that qualitatively, the importance of degree classes are same as in the case of constant $\beta(t),\gamma(t)$. For brevity we have omitted the plots for the ER network.

\subsection{Joint Optimization of Seed and Resource Allocation}
\label{sec:result_joint_seed_res_alloc}

\begin{figure}[ht!]
\centering
\subfloat[ER, $B_{i_0}=0.01$, $\beta=0.07$ \label{fig:seed_norm_res_i0_pt01_ER}]{
\includegraphics[width=41.5mm]{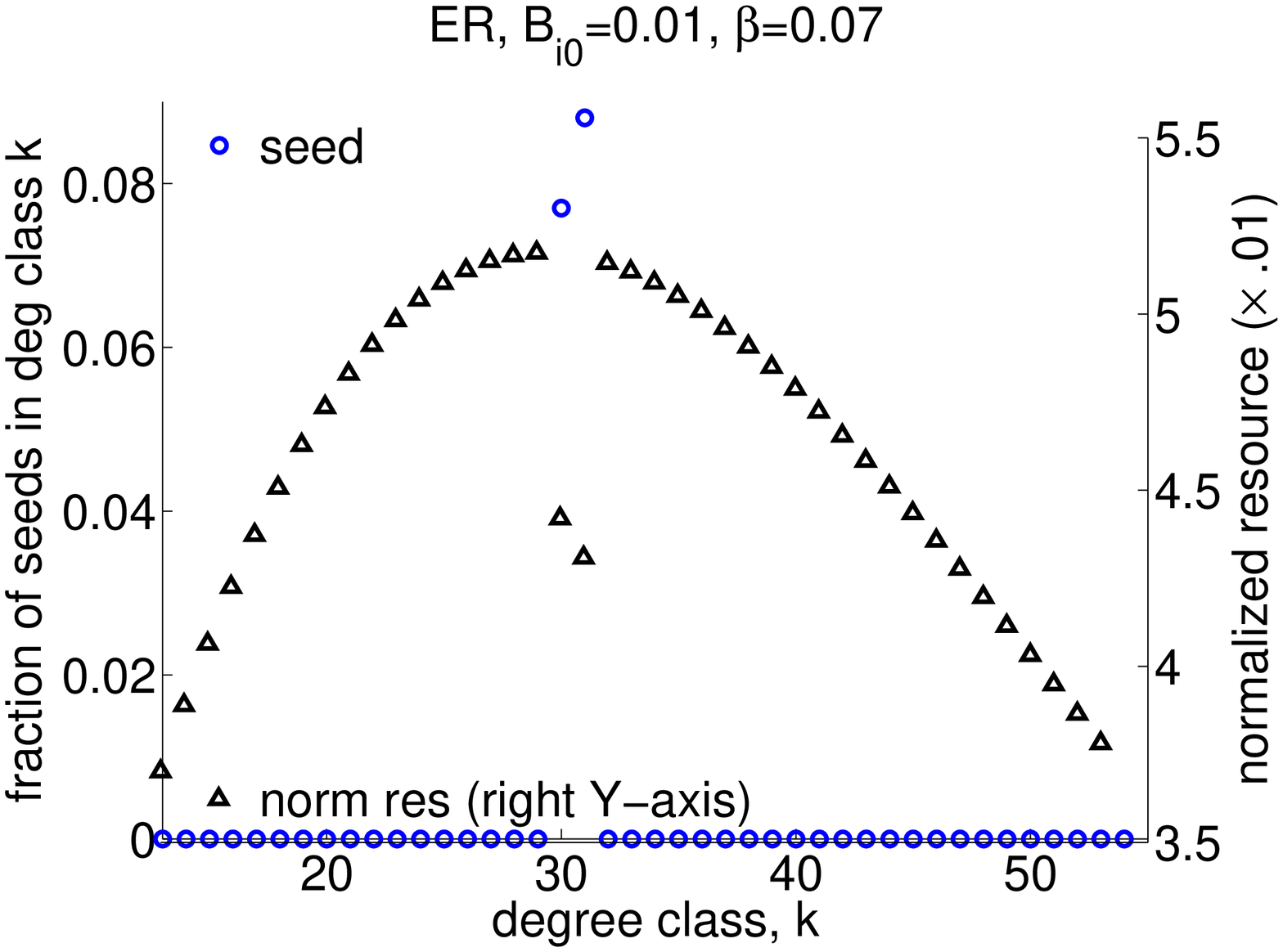} }
\hfill
\subfloat[PL2, $B_{i_0}=0.01$, $\beta=0.07$ \label{fig:seed_norm_res_i0_pt01_PL2}]{
\includegraphics[width=41.5mm]{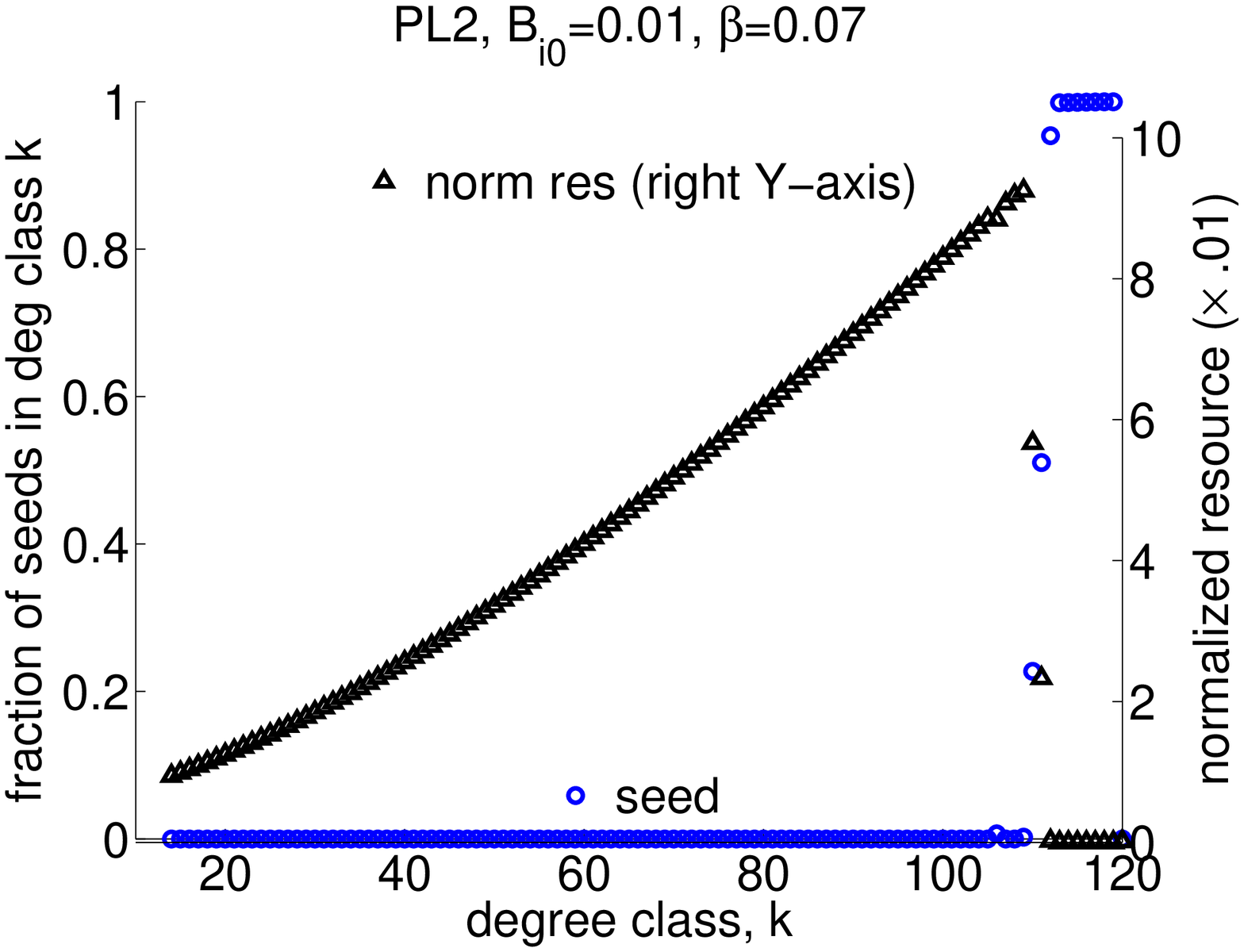} }

\subfloat[ER, $B_{i_0}=0.5$, $\beta=0.07$ \label{fig:seed_norm_res_i0_pt5_ER}]{
\includegraphics[width=41.5mm]{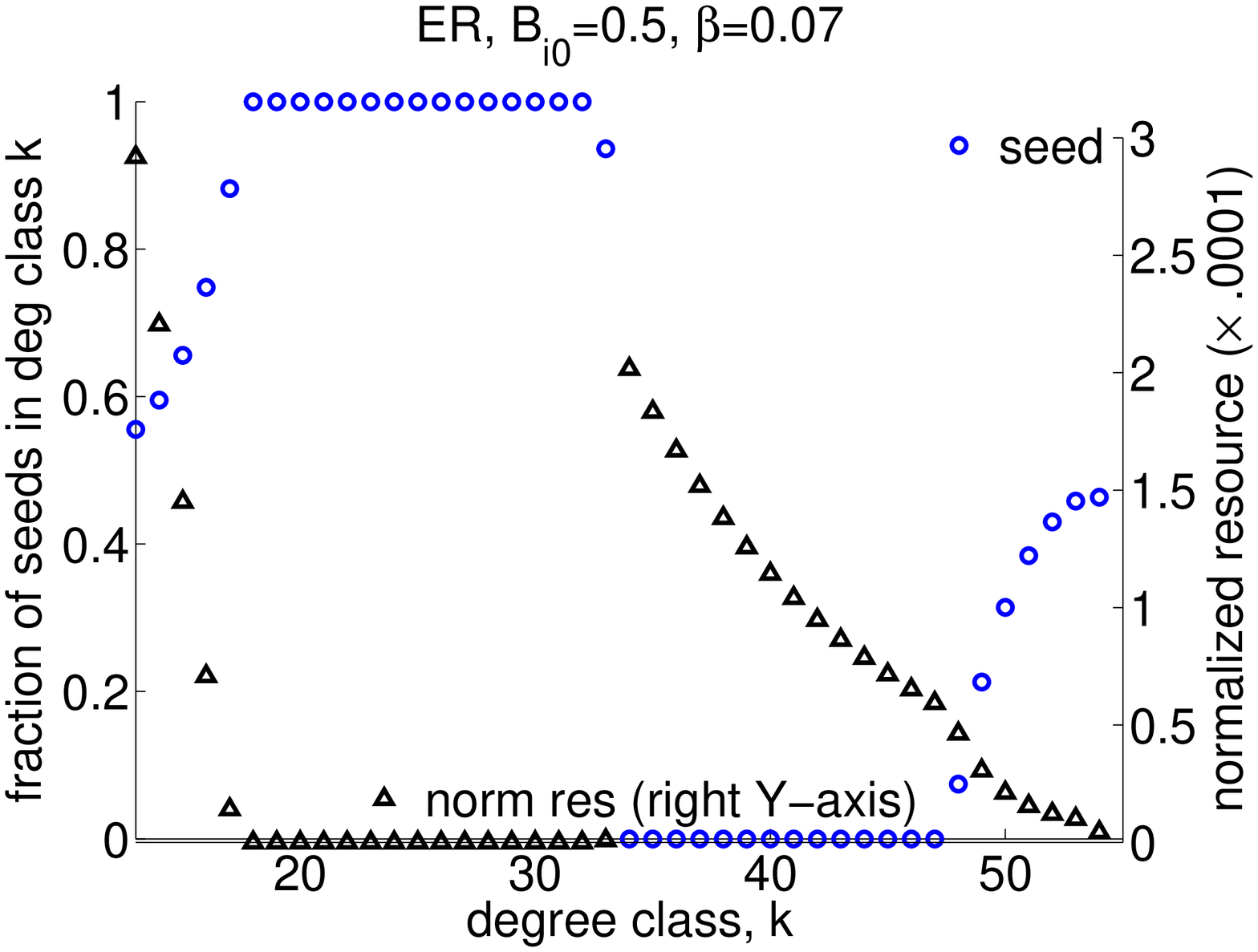} }
\hfill
\subfloat[PL2, $B_{i_0}=0.5$, $\beta=0.07$ \label{fig:seed_norm_res_i0_pt5_PL2}]{
\includegraphics[width=41.5mm]{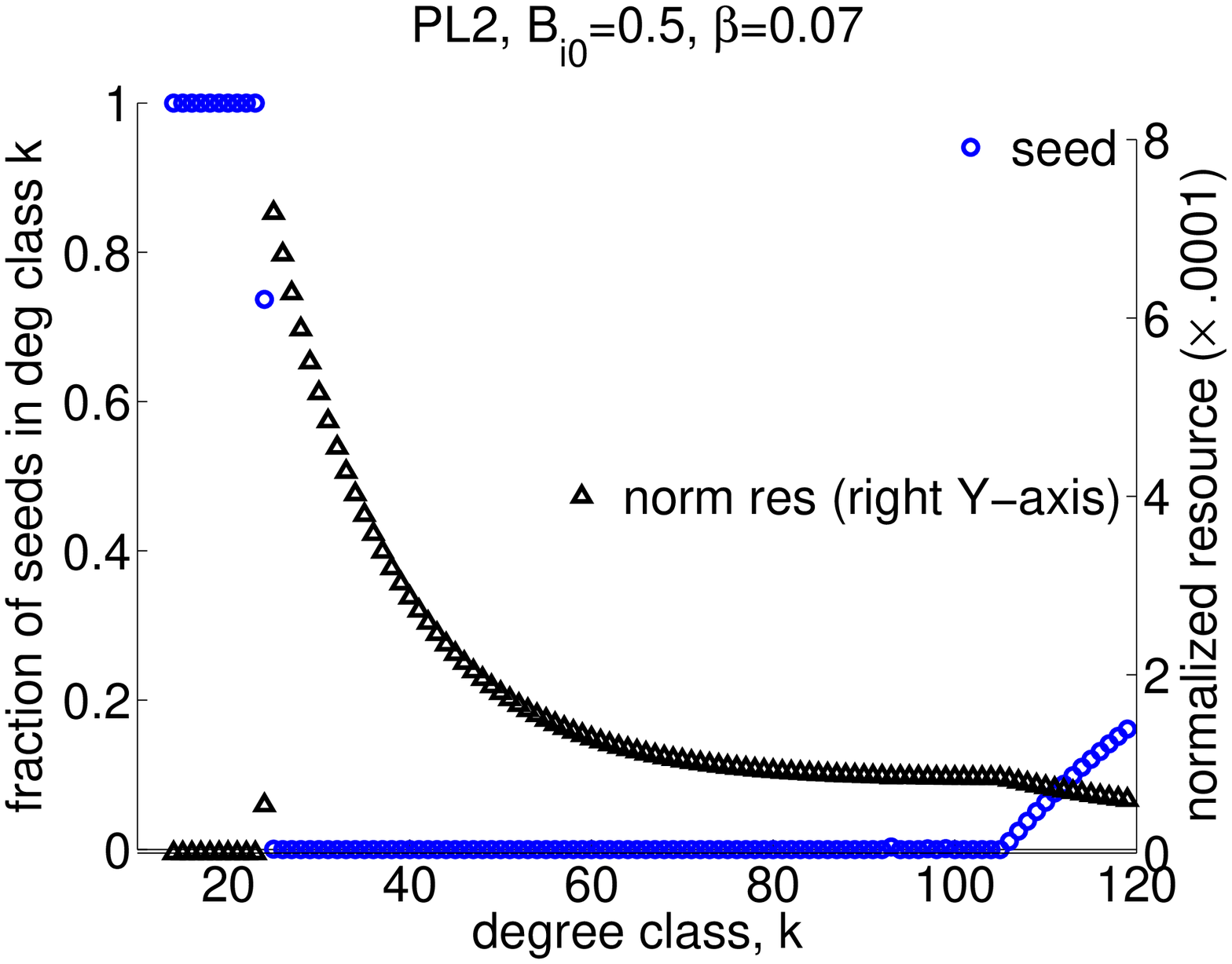} }

\subfloat[ER, $B_{i_0}=0.01$, $\beta=0.18$ \label{fig:seed_norm_res_i0_pt01_beta_pt18_ER}]{
\includegraphics[width=41.5mm]{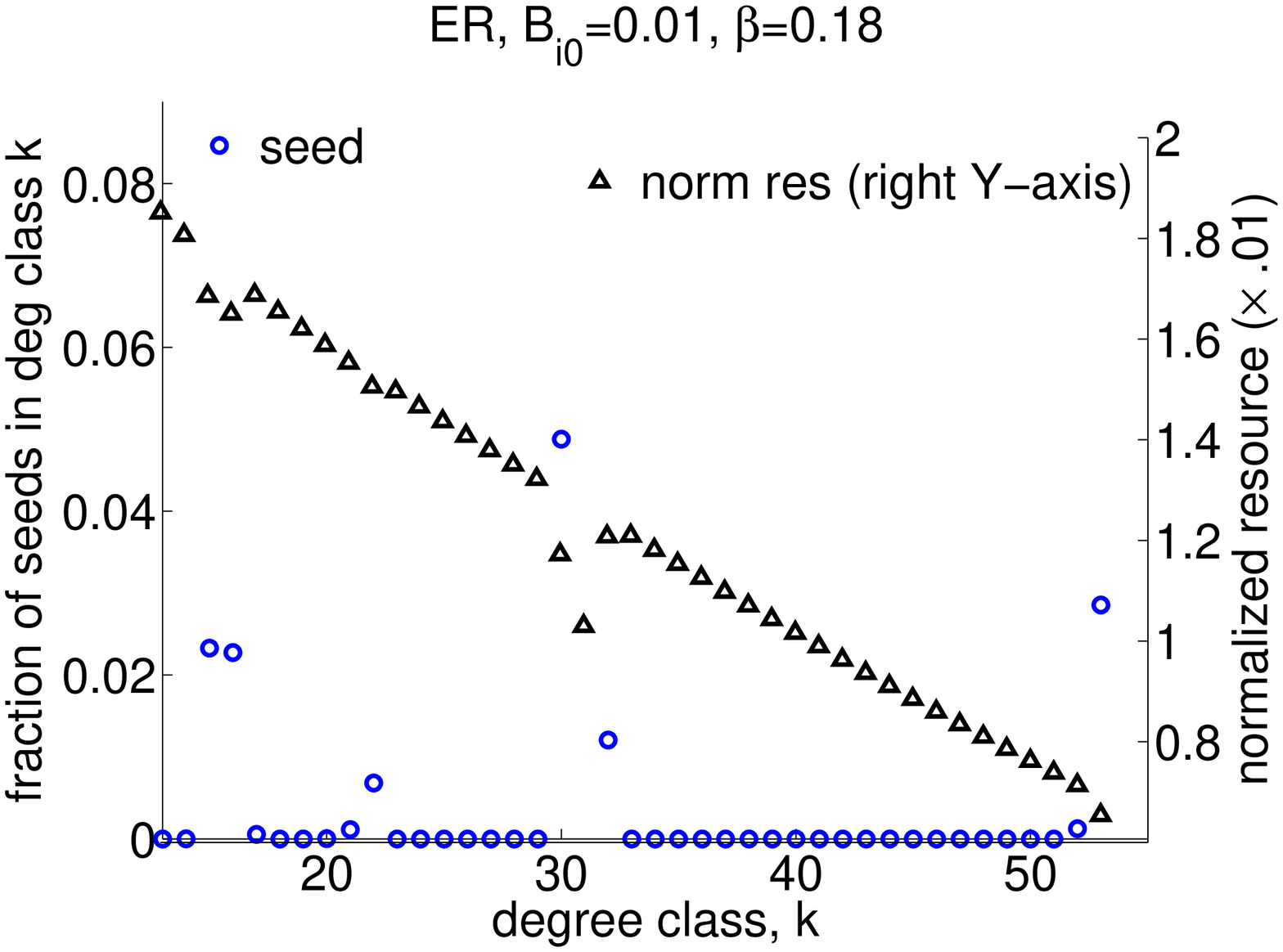} }
\hfill
\subfloat[PL2, $B_{i_0}=0.01$, $\beta=0.18$ \label{fig:seed_norm_res_i0_pt01_beta_pt18_PL2}]{
\includegraphics[width=41.5mm]{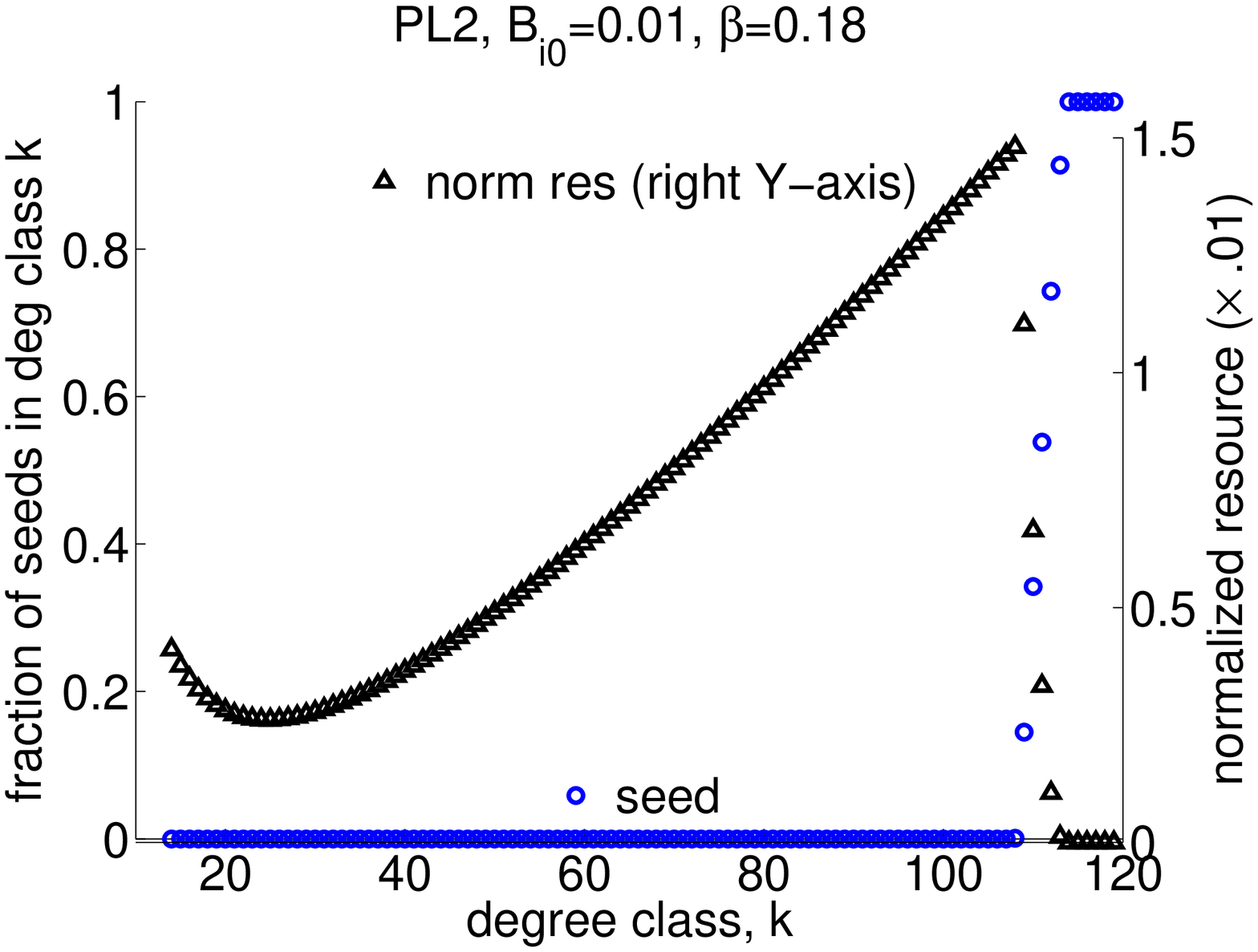} }
\caption{\small{Seed and normalized resource allocation (defined in Eq. (\ref{eq:norm_resource})) for $b=25$, $\beta(t)=\beta=\gamma(t)/10$.}}
\label{fig:seed_norm_res}
\end{figure}

\begin{figure}[ht!]
\centering
\subfloat[ER, $B_{i_0}=0.01$, $\beta(t)$ \label{fig:seed_norm_res_beta1_i0_pt01_ER}]{
\includegraphics[width=41.5mm]{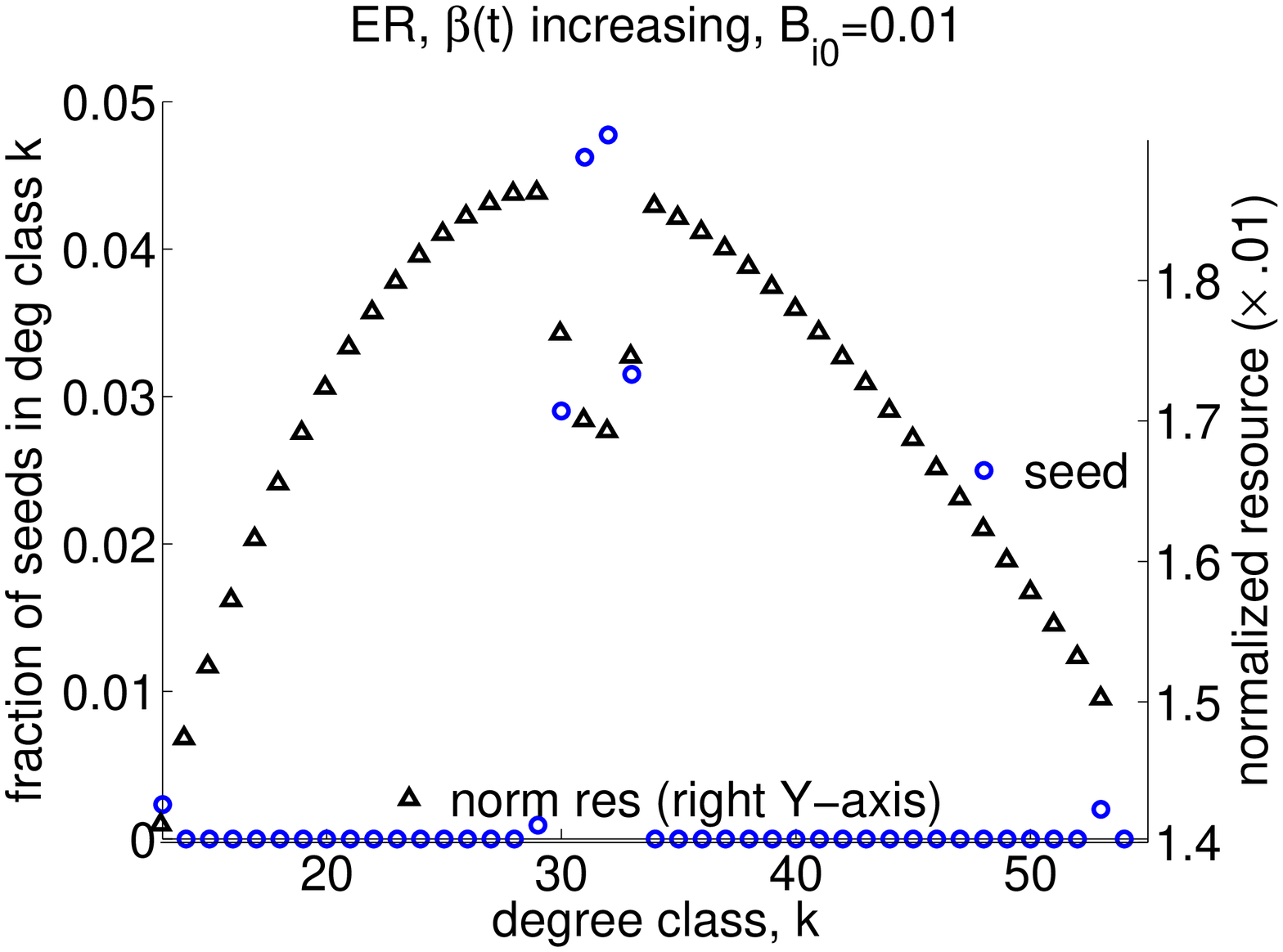} }
\hfill
\subfloat[PL2, $B_{i_0}=0.01$, $\beta(t)$ \label{fig:seed_norm_res_beta1_i0_pt01_PL2}]{
\includegraphics[width=41.5mm]{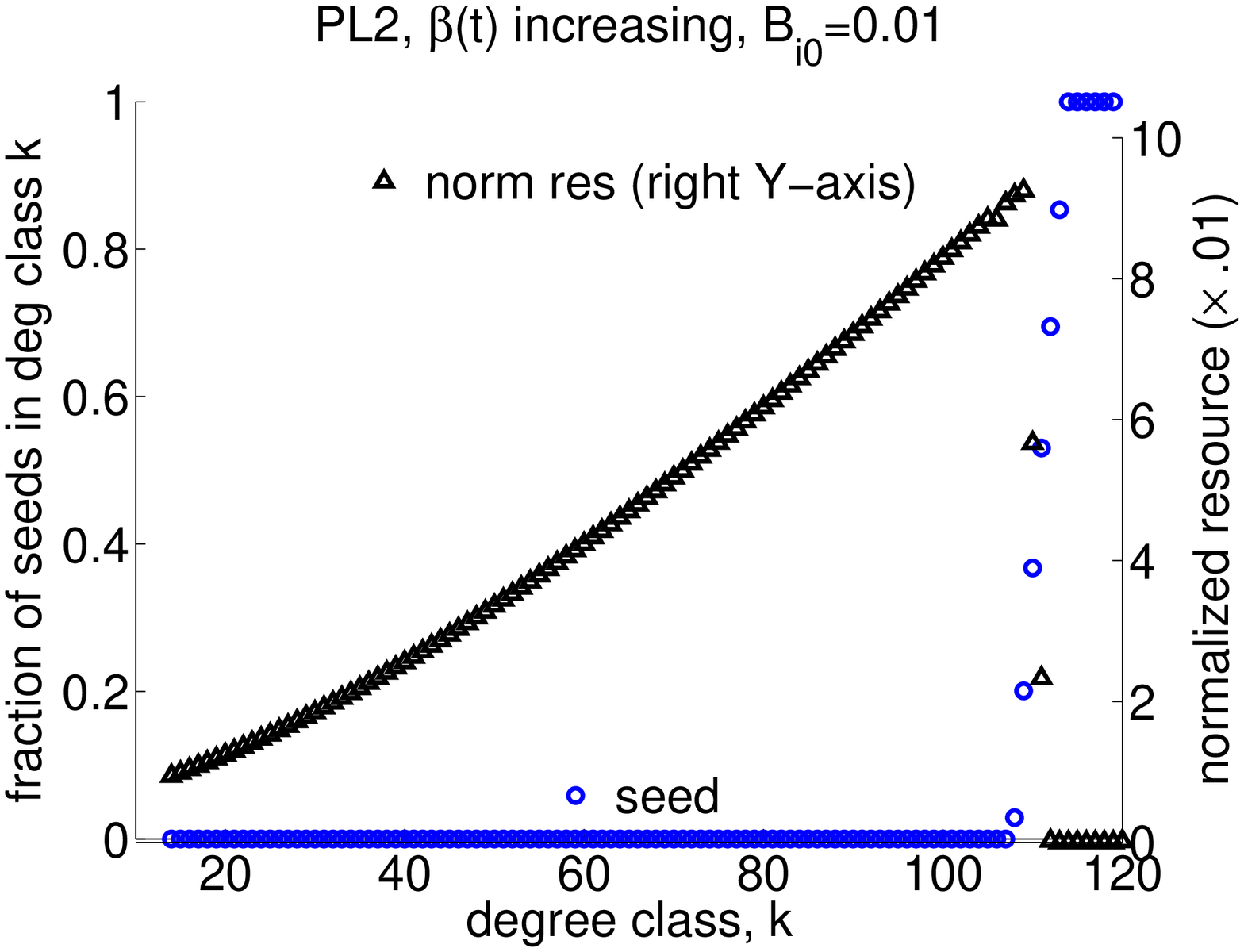} }

\subfloat[ER, $B_{i_0}=0.5$, $\beta(t)$ \label{fig:seed_norm_res_beta1_i0_pt5_ER}]{
\includegraphics[width=41.5mm]{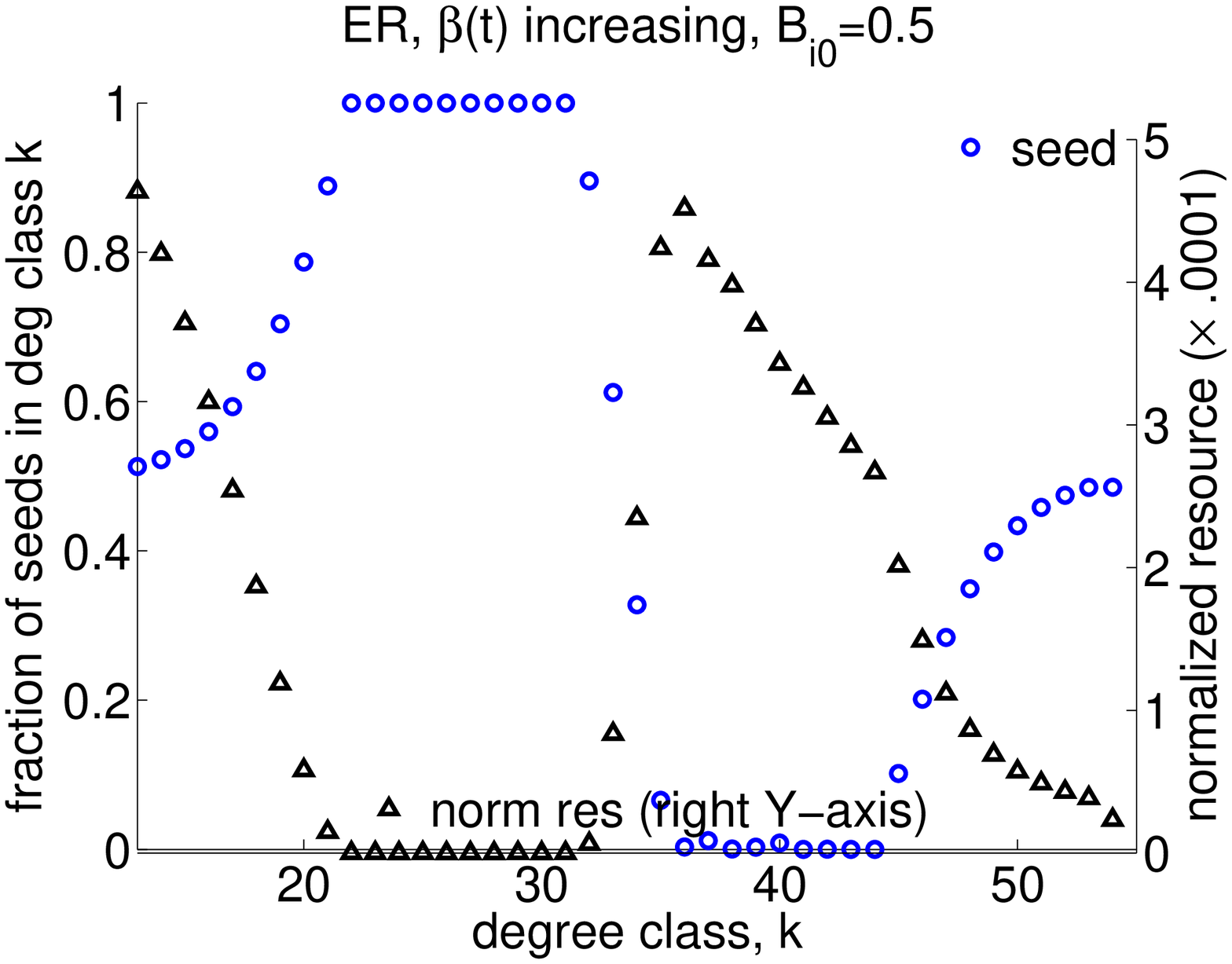} }
\hfill
\subfloat[PL2, $B_{i_0}=0.5$, $\beta(t)$ \label{fig:seed_norm_res_beta1_i0_pt5_PL2}]{
\includegraphics[width=41.5mm]{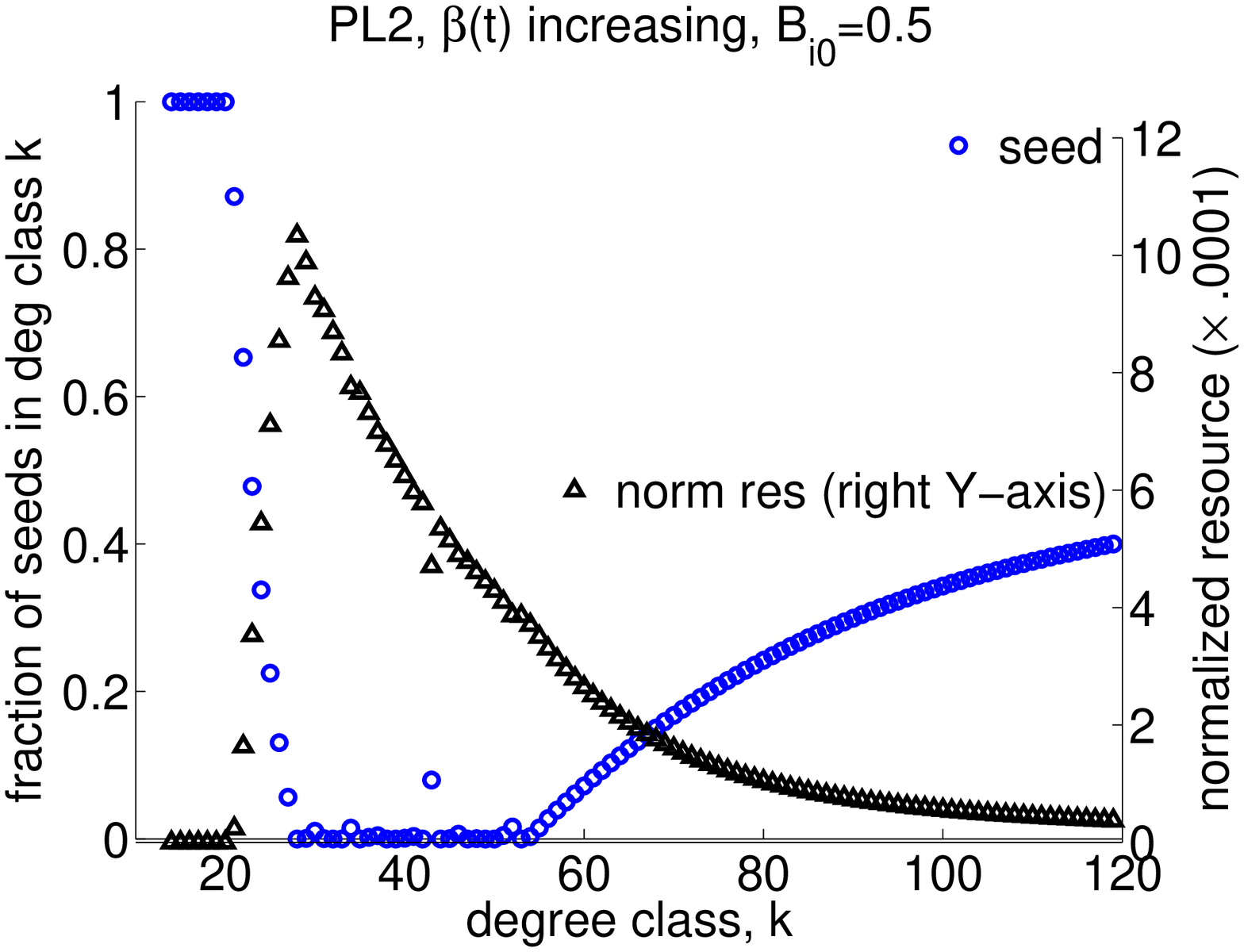} }

\subfloat[ER, $B_{i_0}=0.01$, $3\beta(t)$ \label{fig:seed_norm_res_3beta1_i0_pt01_ER}]{
\includegraphics[width=41.5mm]{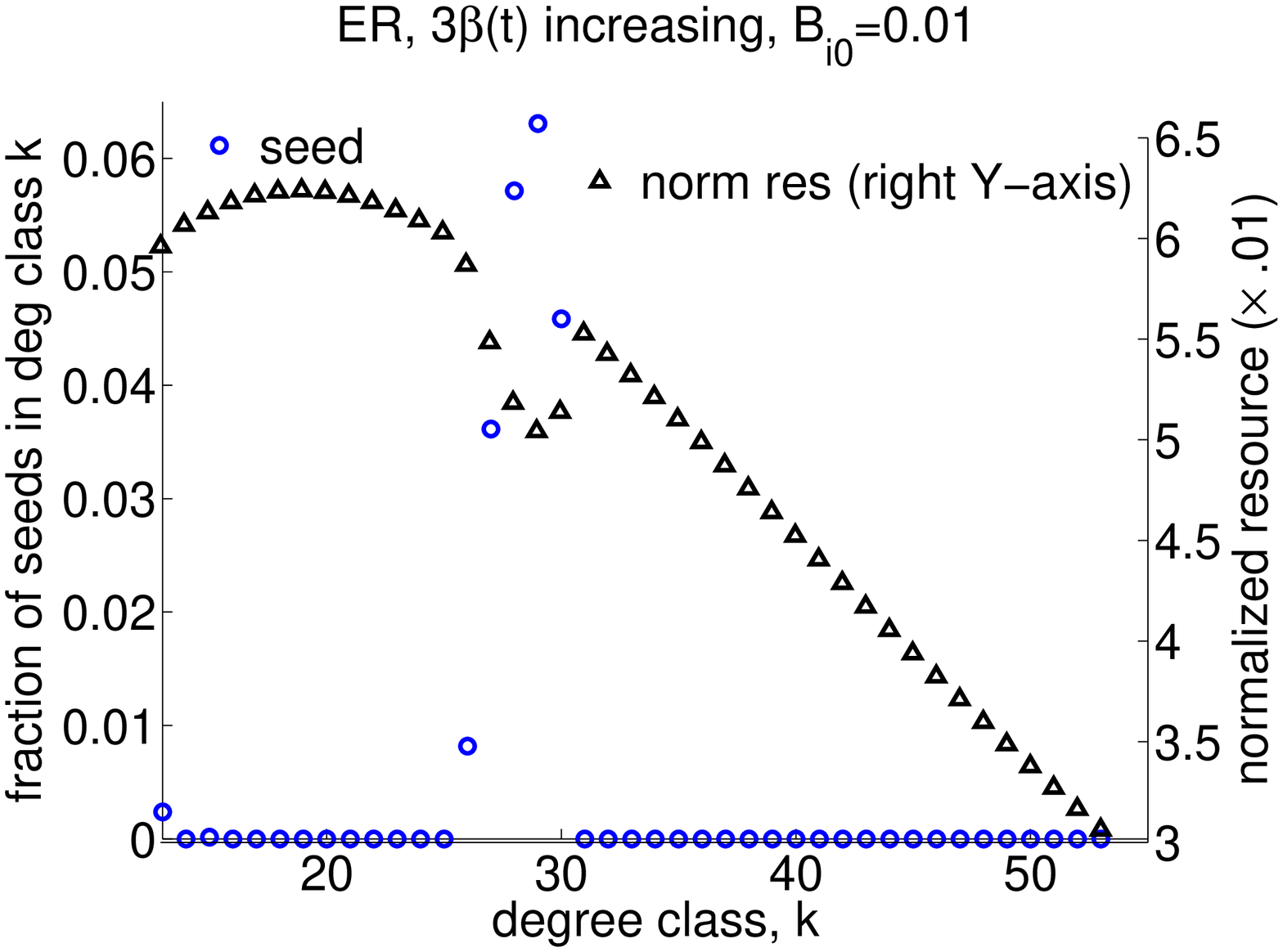} }
\hfill
\subfloat[PL2, $B_{i_0}=0.01$, $3\beta(t)$ \label{fig:seed_norm_res_3beta1_i0_pt01_PL2}]{
\includegraphics[width=41.5mm]{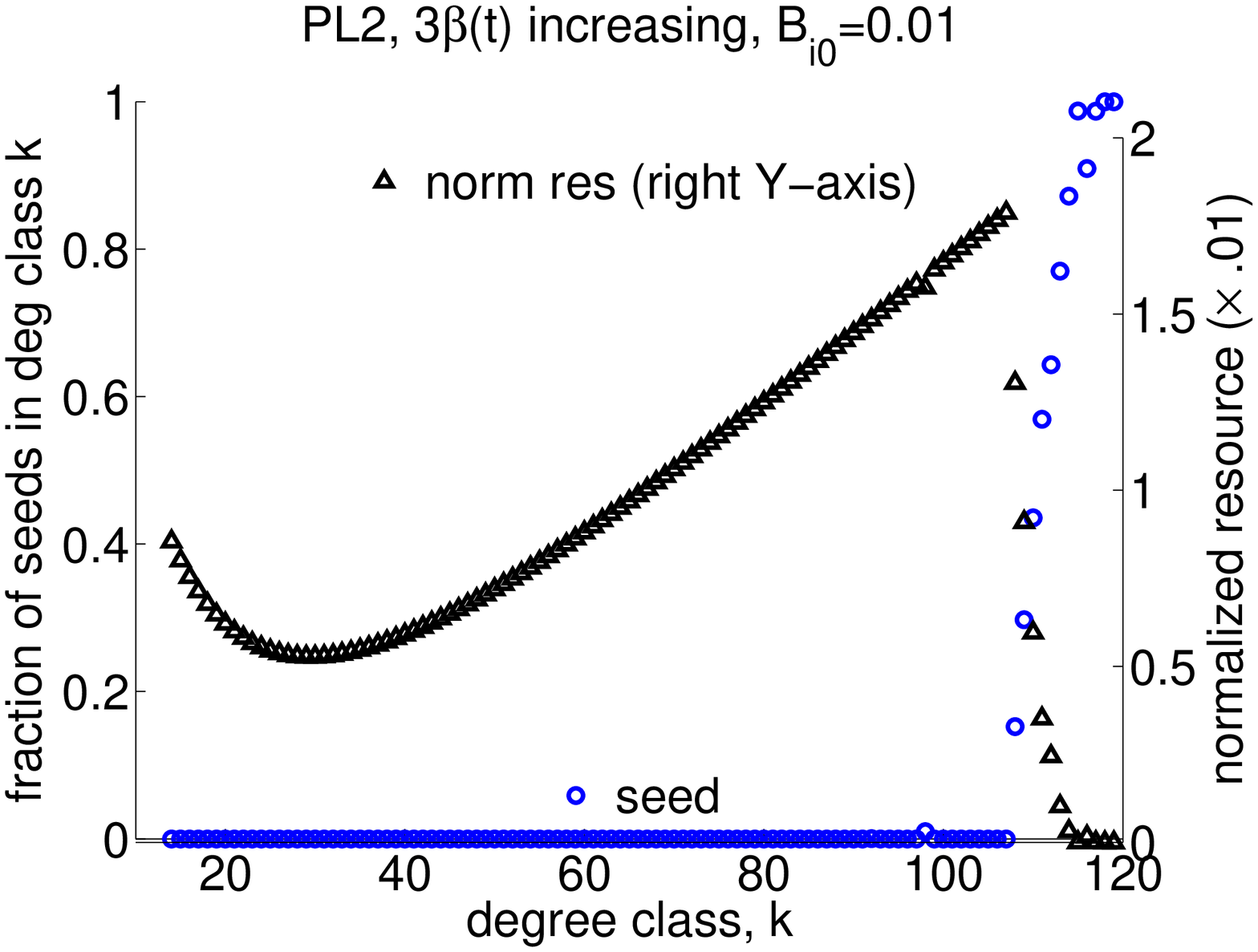} }
\caption{\small{Seed and normalized resource allocation (defined in Eq. (\ref{eq:norm_resource})) for $b=25$, increasing $\beta(t)$ as shown in Fig. \ref{fig:control_norm_res_beta1_b25_PL2}.}}
\label{fig:seed_norm_res_beta1}
\end{figure}

We now discuss the solution to Problem (\ref{eq:opt_prob}) where optimal seed and time varying resource allocation are jointly computed. One expects behavior similar to that in Sec. \ref{sec:result_shape_control_imp_deg_classes}. When the seed budget is low, $B_{i_0}=0.01$, optimal seed allocation and normalized resource allocated to the degree classes are shown in Figs. \ref{fig:seed_norm_res_i0_pt01_ER} and \ref{fig:seed_norm_res_i0_pt01_PL2} for ER and PL2 networks for the case: $\beta(t)=\beta=0.07, \gamma(t)=10\times\beta(t)$. For the scale-free networks, the optimal solution is to target high degree classes as seeds at $t=0$ because they are the best spreaders. The optimal control need not target those degree classes (because they are already infected); thus, degree classes with largest degrees from the remaining ones are preferred. In the ER, network medium degree classes are the preferred seeds.

When the seed budget is increased to $B_{i_0}=0.5$, we see an allocation similar to that in the abundant resource case in Sec. \ref{sec:result_shape_control_imp_deg_classes} (Figs. \ref{fig:seed_norm_res_i0_pt5_ER} and \ref{fig:seed_norm_res_i0_pt5_PL2}). The focus of the optimal solution is not to target best quality spreaders but to target the nodes which are at a disadvantage in receiving the message and directly put them in the infected class. Thus we see a lot of low degrees being targeted in both the networks as seeds and from optimal controls.

When the spreading rate increases to $\beta=0.18$ (Figs. \ref{fig:seed_norm_res_i0_pt01_beta_pt18_ER} and \ref{fig:seed_norm_res_i0_pt01_beta_pt18_PL2}) we again see the disadvantaged lower degree classes attracting more per capita resource than the case when spreading rate was lower ($\beta=0.07$).

Fig. \ref{fig:seed_norm_res_beta1} shows the same result for the case of time varying (increasing) $\beta(t)$ shown in Fig. \ref{fig:control_norm_res_beta1_b25_PL2} and $\gamma(t)=10\times\beta(t)$. Again, the degree classes are targeted as seeds and from controls in a similar manner as in the constant $\beta(t),\gamma(t)$ case.

\subsection{Effect of System Parameters in Problem (\ref{eq:opt_prob})}
\label{sec:result_effect_parameter_no_budget}

In Figs. \ref{fig:J_vs_b}, \ref{fig:J_vs_beta} and \ref{fig:J_vs_i0}, we study the effect of model parameters on optimal reward functions in (\ref{eq:cost_funtion}). The results are compared with the cases where no controls are used, and when the two heuristic control strategies explained in Sec. \ref{sec:heuristic_controls} are used. The curves corresponding to the case when seeds are uniformly selected from the population and are not optimization variable are referred to as `optUniSeed' in the figures.

\subsubsection{Effect of the Cost of Application of Controls}

\begin{figure}[ht!]
\centering
\includegraphics[width=77mm]{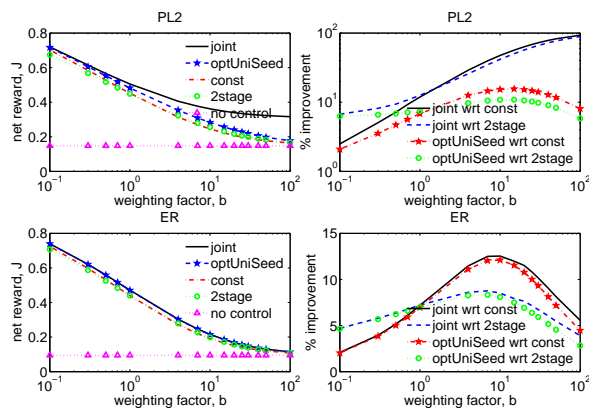}
\caption{\small{Reward $J$ (defined in (\ref{eq:cost_funtion})) vs. weighting parameter $b$, for $B_{i_0}=0.01$ (for Problem (\ref{eq:opt_prob})), $\beta(t)=\beta=0.07=\gamma(t)/10$, $i_{0k}=i_0=0.01,~\forall k\in\mathbb K$ (when seed is not optimized).}}
\label{fig:J_vs_b}
\end{figure}

Fig. \ref{fig:J_vs_b} plots the reward function $J$ with respect to the weighting parameter $b$, which captures the cost of applying control. The larger the value of $b$, the costlier the control becomes. The figure also plots the percentage improvement in the optimal reward function, which the solutions to Problem (\ref{eq:opt_prob}) achieve, over the reward function achieved by the heuristic strategies (in the right panels).

As expected, larger $b$ leads to lower reward. In the case of scale-free networks, the joint optimization leads to much better improvement over just the optimal control problem (compared to the case of ER network). Also, if the resource is too cheap, optimal strategies do not provide any significant improvement over heuristic strategies as both of them reach large fractions of the population. If the resource is too costly, joint optimization leads to much better performance compared to others. In case of the ER network, the percentage improvement achieved by optimal strategies is low; also joint optimization does not offer significant improvement over only optimal resource allocation. This is because network nodes are behaviorally similar in ER network and optimal strategy can exploit their differences to a very limited extent.

\subsubsection{Effect of the Spreading Rate}

\begin{figure}[ht!]
\centering
\includegraphics[width=77mm]{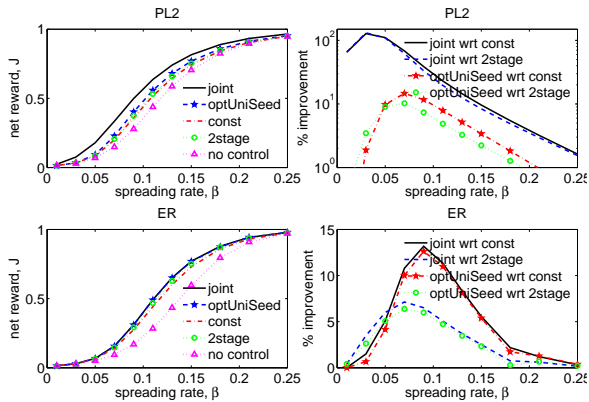}
\caption{\small{Reward $J$ (defined in (\ref{eq:cost_funtion})) vs. spreading rate $\beta(t)=\beta=\gamma(t)/10$, for $B_{i_0}=0.01$ (for Problem (\ref{eq:opt_prob})), $b=25$, $i_{0k}=i_0=0.01,~\forall k\in\mathbb K$ (when seed is not optimized).}}
\label{fig:J_vs_beta}
\end{figure}

Fig. \ref{fig:J_vs_beta} shows the plots of the reward functions with respect to the spreading rate, $\beta(t)=\beta$. Here $\gamma(t)=10\times \beta$. As expected, high values of $\beta$ reduce the importance of campaigning, using both optimal and heuristic strategies, as large fractions of populations can be reached without any effort. As was the case above, in the case of the ER network, joint optimization performs almost the same as only optimal resource allocation.

\subsubsection{Effect of Initial Fraction of Infected Nodes and Seed Budget}

\begin{figure}[ht!]
\centering
\includegraphics[width=77mm]{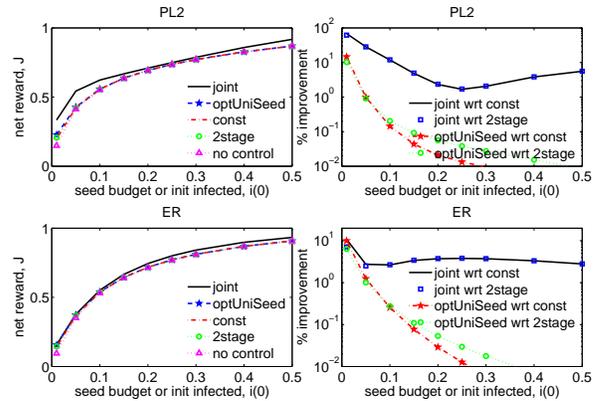}
\caption{\small{Reward $J$ (defined in (\ref{eq:cost_funtion})) vs. initial fraction of infected nodes or seed budget $i(0)$, for $b=25$, $\beta(t)=\beta=0.07=\gamma(t)/10$. For curves other than joint optimization seeds are uniformly selected. On the X-axis, $i(0)=\sum_k p_k i_k(0)=\sum_k p_k i_{0k}$.}}
\label{fig:J_vs_i0}
\end{figure}

Fig. \ref{fig:J_vs_i0} shows the effect of the initial fraction of infected nodes and seed budget on the reward function. For both the networks, optimal resource allocation (without seed optimization) achieves some improvement over heuristic strategies only when there are too few seeds. Joint seed-resource allocation achieves significant improvements for PL2 network for too few or too many seeds. The percentage improvements achieved by the PL2 network are much higher compared to the ER network. Due to the heterogeneous nature of scale-free networks, seed selection is more crucial when few seeds are available. High degree seeds are selected when seed budget is low, as seen in Sec. \ref{sec:result_joint_seed_res_alloc} (and Figs. \ref{fig:seed_norm_res_i0_pt01_PL2}, \ref{fig:seed_norm_res_beta1_i0_pt01_PL2}). In the case of high number of seeds, joint problem allocates seeds to low degree classes (Figs. \ref{fig:seed_norm_res_i0_pt5_PL2}, \ref{fig:seed_norm_res_beta1_i0_pt5_PL2}) thereby directly increasing the reward function. These nodes have fewer links and hence do not receive the message from epidemic spreading efficiently.

\subsection{Effect of System Parameters in the Problem With Budget Constraint (Problem (\ref{eq:opt_prob_budget}))}
\label{sec:result_effect_parameter_budget}

\begin{figure}[ht!]
\centering
\includegraphics[width=77mm]{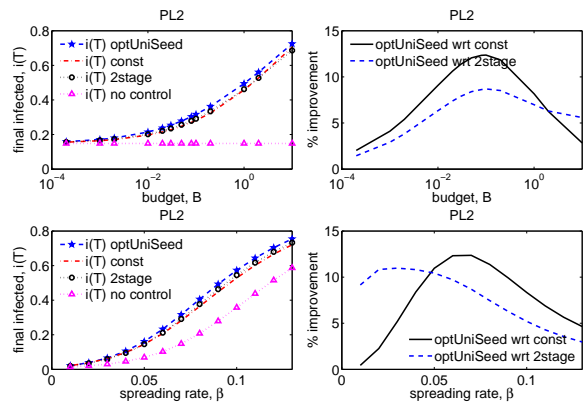}
\caption{\small{Reward $J$ vs. parameters ($b,\beta(t)=\beta=\gamma(t)/10$). Wherever required $B=0.1$, $b=25$, $i_{0k}=i_0=0.01,~\forall k\in\mathbb K$ (for Problem (\ref{eq:opt_prob_budget})).}}
\label{fig:J_vs_param_budget}
\end{figure}

For brevity we show plots for only PL2 network. Fig. \ref{fig:J_vs_param_budget} shows the variations in the reward function (\ref{eq:cost_funtion_budget}) with respect to system parameters in Problem (\ref{eq:opt_prob_budget}). Recall that this is a resource allocation problem with fixed campaigning resources. The reward is simply the fraction of the infected population at the campaign deadline. For these results, we assume that seeds are fixed and are selected uniformly among all degree classes, \emph{i.e.}, $i_{0k}=i_0,~\forall k\in\mathbb K$. The results are compared with the static control strategy and the dynamic two-stage control strategy which uses same budget as the optimal strategy.

Significant percentage improvement is achieved by the optimal resource allocation strategy compared to the heuristic strategies only for intermediate values of the budget. When a lot of resource is available, large fraction of the population is reached even by heuristic strategies and hence not much improvement is achieved by the optimum allocation. Also, very limited resource is not enough to gain any improvement.

The normalized resource allocation in the fixed budget case shows the same qualitative behavior as in Sec. \ref{sec:result_shape_control_imp_deg_classes} (which does not have a budget constraint). For very small budget, in scale-free networks high degree nodes are allocated more per capita resource, followed by medium and low degree nodes. In ER, medium degrees are favored, this behavior being similar to the costly resource case in Sec. \ref{sec:result_shape_control_imp_deg_classes}. If the budget is too high, the trend is the same as in the cheap resource case in Sec. \ref{sec:result_shape_control_imp_deg_classes}. The economic interpretation of multiplier $\mu$ associated with the relaxed constraint is cost per unit resource. Hence, low budget leads to high value of $\mu^*$, the multiplier's value at the optimum, and high budget leads to small value of $\mu^*$, which explains this behavior.  We are omitting the figures for brevity.

\section{Results for Slashdot Social Network}
\label{sec:results_real_nw}

\begin{figure}[ht!]
\centering
\includegraphics[width=53mm]{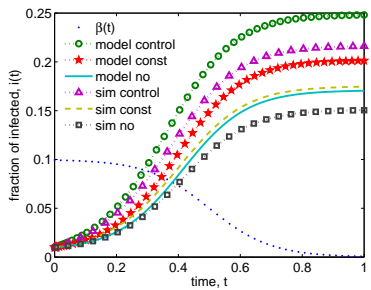}
\caption{\small{Time evolution of $i(t)$ predicted by the model and in simulation on Slashdot social network. $b=25$, $\beta(t)$ as shown above, $\gamma(t)=10\times \beta(t)$, $i_{0k}=0.01~\forall k$.}}
\label{fig:it_vs_t_control_uncon_const_th_sim_beta2}
\end{figure}

\begin{figure}[ht!]
\centering
\includegraphics[width=88mm]{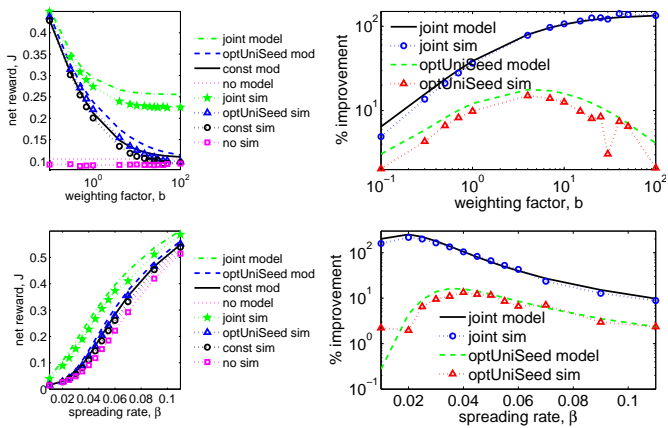}
\caption{\small{Reward $J$ vs. parameters ($b,\beta(t)=\beta=\gamma(t)/10$)---from model and simulation---for the Slashdot network. Whenever required $b=10,\beta(t)=\beta=0.04, \gamma(t)=10\times\beta(t)$, $i_{0k}=0.01\forall k$ when seed is not optimized.}}
\label{fig:J_vs_b_beta_th_sim_actual_nw}
\end{figure}

The degree based compartmental model for SI epidemics presented in Sec. \ref{sec:sys_model_prob_formaulation} assumes that the process is running on an uncorrelated network---an assumption satisfied by configuration model networks but not real networks. In this section, we test the accuracy of the model on a real network via simulations. For this purpose, we use a $4000$ node Slashdot social network obtained from \cite{snap}. The network has a minimum degree of $1$, maximum degree $661$ and mean $30.42$.

Fig. \ref{fig:it_vs_t_control_uncon_const_th_sim_beta2} compares the system evolution predicted by the model and by simulation. The simulation results are averaged over $400$ runs. Optimal controls are computed using the empirical degree distribution of the network (with the implicit assumption that the network is uncorrelated). Seeds are uniformly selected, but seeding is different in each run.

We note in Fig. \ref{fig:it_vs_t_control_uncon_const_th_sim_beta2} that the model overestimates the fraction of infected nodes in the network. This is because real social networks have high levels of clustering---the number of triangles in the network---because two `friends' of a person are also likely to be friends. On the other hand, uncorrelated networks have very low levels of clustering---they are `locally tree like' \cite[Sec. 17.10.1]{newman2009networks}, devoid of short loops, and have many more long edges compared to real social networks. Thus, information diffuses far and wide more quickly in uncorrelated networks than in real networks---which explains the behavior in Fig. \ref{fig:it_vs_t_control_uncon_const_th_sim_beta2}.

Inspite of this inaccuracy in modeling, the controls computed from the model are still useful for the real network. This is demonstrated in Fig. \ref{fig:J_vs_b_beta_th_sim_actual_nw} for a wide range of model parameters. The percentage improvement which the optimal control---with and without seeds as optimization variables---enjoys over the constant (or static) control is similar: whether predicted by model or observed in simulations. The percentage improvements for the net reward functions in the controlled system in the case of simulations are computed with respect to net reward functions obtained when constant controls are used in the simulations.

All simulation results are averaged over $400$ runs. In all the runs, the network is the same but seeds are different---either selected uniformly from the population (in the case of no control, constant control and only time varying resource optimization) or, in the case of joint problem, a node in degree class $k$ is selected as seed with probability $i_{0k}$, where $i_{0k}$ is the output of the optimization problem.

\section{Conclusion}
\label{sec:conclusion}

In this work, we have applied techniques from optimal control on a large optimality system for allocating campaigning resources over (i) time and (ii) degree classes for maximizing the spread of a piece of information over social networks. Information dissemination is modeled as a Susceptible-Infected epidemic and direct recruitment of susceptible nodes to the infected class is used as the strategy to enhance information spreading. The seed for the epidemic is also jointly optimized along with time varying resource allocation. The whole network is divided into degree classes based on node degrees and each degree class is influenced by a separate control. The aim is to maximize a (net) reward function, which is a linear combination of reward due to the extent of information spread and cost due to application of controls. We have also studied a variation of the above problem---maximizing information spread under a fixed budget constraint (fraction of seeds is given).

We prove the existence of a solution to the optimal control problem, provide analytical structural results for the shape of the controls, and provide a sufficient condition for uniqueness of solution. We solve the above optimal control problem using Pontryagin's Maximum Principle. Our formulation and system parameters lead to large optimality systems with over $200$ differential equations. Numerical schemes use the forward-backward sweep technique and its variations to solve different problems studied in this paper. These schemes are more efficient than direct conversion of the optimal control problems to non-linear optimization problems. We analyse the convergence of the forward-backward sweep technique for our system.

We compare the optimal results with two heuristic strategies. The first one is constant throughout the campaign horizon, and second is dynamic in nature and is active only during initial periods of the campaign (which was found to be more important for information spreading). Results show significant gains over these non-optimal strategies.

We also found that per capita resource allocation over the degree classes varies depending on the network topology and system parameters such as cost of the resource and spreading rate. For example, if resource is costly (scarce), medium degree nodes are allocated more resources in Erd\H os-R\'enyi networks, but higher degree nodes are favored in the case of scale-free networks. If the resource is cheap (abundant), the allocation to low degree nodes is more than that to medium degree nodes in both networks.

We tested the accuracy of the degree based compartmental model for SI epidemics on a real social network by simulating various control strategies. Although the fraction of infected nodes is slightly overestimated by the model, the performance improvements achieved by the optimal control over the constant control strategy is almost the same in the simulations as predicted by the model.

\footnotesize
\bibliographystyle{IEEEtran}
\bibliography{bibliography_database}

%
%
%

\newpage
\normalsize

\twocolumn[
\centering
\begin{@twocolumnfalse}
\textbf{\LARGE{SUPPLEMENTARY MATERIAL \\ \vspace{.2cm} Optimal Resource Allocation Over Time and Degree Classes for Maximizing Information Dissemination in Social Networks}} \\
\vspace{.3cm}
\large{Kundan Kandhway and Joy Kuri}
\vspace{.6cm}
\end{@twocolumnfalse}
]


\appendices
\section{Proof of Theorem \ref{thm:soln_exist}: Existence of a Solution}
\label{app:proof_existence}
Let the system of ODEs represented by (\ref{eq:opt_prob_states}) be denoted by $\frac{d}{dt} {\boldsymbol i}(t)=\boldsymbol f(\boldsymbol u(t), \boldsymbol i(t))$, where the RHS is a vector function with $|\mathbb K|$ elements whose components are given by (\ref{eq:opt_prob_states}) and $\boldsymbol i(t)=\{i_k(t),~k\in\mathbb K\}$. We make use of Cesari's theorem \cite[Ch. 3]{fleming1975deterministic} in this proof. We will use the vector 1-norm throughout this proof. However, note that all $p$-norms ($p\geq 1$) for vectors are equivalent\footnote{R. A. Horn and C. R. Johnson, \emph{Matrix Analysis}. Cambridge university press, 1990 (Sec. 5.4.7).}; hence the result holds irrespective of the norm used. Cesari's theorem states that an optimal control problem has a solution if the following conditions are satisfied:
\begin{enumerate}
\item $|\boldsymbol f(\boldsymbol u(t), \boldsymbol i(t))|\leq C_1(1+|\boldsymbol i(t)|+|\boldsymbol u(t)|)$, for $C_1>0$. Using 1-norm for vectors, $|\boldsymbol f(\boldsymbol u(t), \boldsymbol i(t))|\leq \text{max}_{t\in[0,T]}\{\beta(t),\gamma(t)\} K_{max}^2 (|\boldsymbol i(t)|+|\boldsymbol u(t)|)$.

\item $|\boldsymbol f(\boldsymbol u(t), \boldsymbol i(t))-\boldsymbol f(\boldsymbol u(t), \hat {\boldsymbol i}(t))|\leq C_2|\boldsymbol i(t)-\hat{\boldsymbol i}(t)|(1+|\boldsymbol u(t)|)$, for $C_2>0$. Evaluating the left hand side (we make use of $\hat s_k(t)=1-\hat i_k(t)$),
\small
\begin{align*}
& |\boldsymbol f(\boldsymbol u(t), \boldsymbol i(t))-\boldsymbol f(\boldsymbol u(t), \hat {\boldsymbol i}(t))| \\
= & \sum_{k\in\mathbb K}\bigg|\beta(t) ks_k(t) \sum_{l\in \mathbb K} \left(q_{l} i_l(t)\right) - \beta(t) k\hat s_k(t) \sum_{l\in \mathbb K} (q_{l} \hat i_l(t)) \\
& + \gamma(t)u_k(t)s_k(t) - \gamma(t)u_k(t)\hat s_k(t)\bigg| \\
\leq & \sum_{k\in\mathbb K} \bigg|\beta(t) ks_k(t) \sum_{l\in \mathbb K} \left(q_{l} i_l(t)\right) - \beta(t) k\hat s_k(t) \sum_{l\in \mathbb K} \left(q_{l} i_l(t)\right) \\
& + \beta(t) k\hat s_k(t) \sum_{l\in \mathbb K} \left(q_{l} i_l(t)\right) - \beta(t) k\hat s_k(t) \sum_{l\in \mathbb K} (q_{l} \hat i_l(t))\bigg| \\
& + \sum_{k\in\mathbb K}\bigg|\gamma(t)u_k(t)s_k(t) - \gamma(t)u_k(t)\hat s_k(t)\bigg| \\
\leq & \sum_{k\in\mathbb K} \underset{t\in[0,T]}{\text{max}}\{\beta(t)\} K_{max} \bigg|\sum_{l\in \mathbb K} (q_li_l(t))\bigg|.\bigg| s_k(t)-\hat s_k(t) \bigg|\\
& + \sum_{k\in\mathbb K} \underset{t\in[0,T]}{\text{max}}\{\beta(t)\} K_{max} |\hat s_k(t)|. \sum_{l\in \mathbb K} \Big(q_l\big|i_l(t)-\hat i_l(t)\big|\Big)\\
& + \sum_{k\in\mathbb K} \underset{t\in[0,T]}{\text{max}}\{\gamma(t)\} |\boldsymbol u(t)||s_k(t)-\hat s_k(t)| \\
\leq & 2\underset{t\in[0,T]}{\text{max}}\{\beta(t)\} K_{max}^2 |\boldsymbol i(t)-\hat{\boldsymbol i}(t)| \\
& + \underset{t\in[0,T]}{\text{max}}\{\gamma(t)\} |\boldsymbol u(t)|.|\boldsymbol i(t)-\hat{\boldsymbol i}(t)| \\
\leq & \underset{t\in[0,T]}{\text{max}}\{\gamma(t), ~2\beta(t) K_{max}^2\}\times|\boldsymbol i(t)-\hat{\boldsymbol i}(t)|\times(1+|\boldsymbol u(t)|),
\end{align*}
\normalsize
which is as required. We have made use of estimations such as $k\leq K_{max}$, $\big|\sum_{l\in \mathbb K} (q_li_l(t))\big|\leq K_{max}$ (note that $q_l$'s and $i_l$'s are probabilities and hence $\leq 1$), $\sum_{k\in\mathbb K}|\hat s_k(t)|\leq K_{max}$, $|u_k(t)|\leq |\boldsymbol u(t)|$ and $|a+b|\leq |a|+|b|$ for $a,b\in\mathbb R$.

\item The admissible set of controls $U^{|\mathbb K|}$ is non-empty by construction (Definition \ref{def:set_of_admissible_controls}).

\item The control at time $t$, $\boldsymbol u(t)$ takes values in a closed space $\mathbb R^{|\mathbb K|}$. The whole space contains all its limit points and hence it is closed\footnote{W. Rudin, \emph{Principles of Mathematical Analysis.} McGraw-Hill New York, 1964}.

\item The reward due to the terminal state in the reward function $J$ in (\ref{eq:cost_funtion}), $\sum_{k\in\mathbb K}p_ki_k(T)$, takes values in a compact space [0,1] and $\sum_{k\in\mathbb K}p_ki_k(T)$ is continuous in $i_k(T)$.

\item $\mathbb R^{|\mathbb K|}$ is a convex space, $\boldsymbol f(\boldsymbol u(t), \boldsymbol i(t))$ is linear in $\boldsymbol u(t)$ and $\sum_{k\in \mathbb K}g_k(u_k(t))$ is convex in $\boldsymbol u(t)$ (Assumption \ref{assumption:gk_increasing}).

\item The final requirement of the theorem is that, $\exists$ a continuous function $\sigma(\boldsymbol u(t))$, such that, $\sum_{k\in \mathbb K}g_k(u_k(t))\geq \sigma(\boldsymbol u(t))$ and $\frac{\sigma(\boldsymbol u(t))}{|\boldsymbol u(t)|}\rightarrow\infty$ as $|\boldsymbol u(t)|\rightarrow\infty$. Choose $\sigma(\boldsymbol u(t)) = \sum_{k\in \mathbb K}g_k(u_k(t))$. Now, $|\boldsymbol u(t)|\rightarrow\infty$ means either the largest component $u_p(t)\rightarrow\infty$ or, the smallest component $u_q(t)\rightarrow-\infty$. In the former case, $\underset{|\boldsymbol u(t)|\rightarrow\infty}{\lim}\frac{\sum_{k\in \mathbb K}g_k(u_k(t))}{|\boldsymbol u(t)|}>\underset{u_p(t)\rightarrow\infty}{\lim}\frac{g_p(u_p(t))}{|\mathbb K|.u_p(t)}=\underset{u_p(t)\rightarrow\infty}{\lim}\frac{g_p'(u_p(t))}{|\mathbb K|}\rightarrow\infty$. We use L'Hospital's rule. Since $g_p(.)$ is strictly convex and $g_p(0)=0$, so $g_p''(.)>0\Rightarrow g_p'(.)$ is strictly increasing for positive arguments. In the latter case, $\underset{|\boldsymbol u(t)|\rightarrow\infty}{\lim}\frac{\sum_{k\in \mathbb K}g_k(u_k(t))}{|\boldsymbol u(t)|}>\underset{u_q(t)\rightarrow-\infty}{\lim}\frac{g_q(u_q(t))}{-|\mathbb K|.u_q(t)}=\underset{u_q(t)\rightarrow-\infty}{\lim}\frac{g_q'(u_q(t))}{-|\mathbb K|}\rightarrow\infty$. Note that by Assumption \ref{assumption:gk_even}, $g_q(.)$ is an even function, so strictly decreasing for negative arguments.

\end{enumerate}

\section{Proof of Theorem \ref{thm:conv_fw_bk_sweep}:Convergence of Forward-Backward Sweep Algorithm}
\label{app:proof_conv_fw_bk_sweep}

We use the techniques in \cite{mcasey2012convergence} for the analysis in this section. In this section we will denote the iteration number by $(n)$. Then forward-backward sweep uses the following iteration:
\begin{eqnarray}
\textnormal{Initialize: } & & u_k^{(0)}. \nonumber \\
\textnormal{Iterate: } & & \nonumber \\
\frac{d}{dt} i_k^{(n+1)}(t) & = & \beta(t) ks_k^{(n+1)}(t) \sum_{l\in \mathbb K} \left(q_{l} i_l^{(n+1)}(t)\right) \nonumber \\
& & + \gamma(t)u_k^{(n)}(t)s_k^{(n+1)}(t);\nonumber \\
i_k^{(n+1)}(0) & = & i_{0k}. \nonumber
\end{eqnarray}
\begin{eqnarray}
\frac{d}{dt} \lambda_k^{(n+1)}(t) & = & \beta(t) k \lambda_k^{(n+1)}(t)\sum_{l\in\mathbb K}(q_li_l^{(n+1)}(t)) \nonumber  \\
& & -\beta(t) q_k \sum_{j\in \mathbb K}(\lambda_j^{(n+1)}(t)js_j^{(n+1)}(t)) \nonumber \\
& & + \gamma(t)u_k^{(n)}(t)\lambda_k^{(n+1)}(t);\nonumber  \\
\lambda^{(n+1)}(T)& = & p_k. \nonumber \\
u_k^{(n+1)} & = & \frac{\gamma(t)}{2c_k}s_k^{(n+1)}(t))\lambda_k^{(n+1)}(t). \label{eq:u_for_quard_cost}
\end{eqnarray}
The state/adjoint variables in the $n$th iteration satisfy the above equations. The solutions, $\boldsymbol i(t)=\{i_k(t),~k\in\mathbb K\}$ and $\boldsymbol \lambda(t)=\{\lambda_k(t),~k\in\mathbb K\}$ satisfy Eqs. (\ref{eq:opt_prob_states}) and (\ref{eq:costate_diff_eq}) respectively.

Let the errors in $n$th iteration be denoted by $e_{i_k}^{(n)}(t) = i_k(t)-i_k^{(n)}(t)$, $e_{\lambda_k}^{(n)}(t) = \lambda_k(t)-\lambda_k^{(n)}(t)$ and $e_{u_k}^{(n)}(t) = u_k(t)-u_k^{(n)}(t)$ with vectors of errors represented by $\boldsymbol e_{i}^{(n)}(t)$, $\boldsymbol e_{\lambda}^{(n)}(t)$ and $\boldsymbol e_{u}^{(n)}(t)$ respectively. The error $e_{i_k}^{(n+1)}(t)$ evolves as:
\begin{eqnarray*}
& \frac{d}{dt}e_{i_k}^{(n+1)}(t) \hspace{-1em} & = \beta(t) k \big[ s_k(t)\Sigma_l(q_l i_l(t)) \\
& & - s_k^{(n+1)}(t)\Sigma_l(q_l i_l^{(n+1)}(t)) \big] + \gamma(t) \big[ u_k(t)s_k(t) \\
& & - u_k^{(n)}(t)s_k^{(n+1)}(t) \big]; \\
& e_{i_k}^{(n+1)}(0) & = 0.
\end{eqnarray*}
This leads to:
\begin{align}
& e_{i_k}^{(n+1)}(t) = \int_0^t \Big\{ \beta(\tau) k \big[ s_k(\tau)\Sigma_l(q_l i_l(\tau)) \nonumber\\
& - s_k^{(n+1)}(\tau)\Sigma_l(q_l i_l^{(n+1)}(\tau)) \big] + \gamma(\tau) \big[ u_k(\tau)s_k(\tau) \nonumber\\
& - u_k^{(n)}(\tau)s_k^{(n+1)}(\tau) \big] \Big\} d\tau. \nonumber\\
\Rightarrow & |e_{i_k}^{(n+1)}(t)| \leq \int_0^t \Big\{ \big|\beta(\tau) k \big[ s_k(\tau)\Sigma_l(q_l i_l(\tau)) \nonumber\\
& - s_k^{(n+1)}(\tau)\Sigma_l(q_l i_l^{(n+1)}(\tau)) \big]\big| + \big|\gamma(\tau) \big[ u_k(\tau)s_k(\tau) \nonumber\\
& - u_k^{(n)}(\tau)s_k^{(n+1)}(\tau) \big]\big| \Big\} d\tau. \label{eq:conv_mod_eik}
\end{align}
The first of the two modulus terms in the RHS of (\ref{eq:conv_mod_eik}) can be estimated further as:
\begin{align}
& \hspace{-.3cm} \int_0^t \big|\beta(\tau) k \big[ s_k(\tau)\underset{l}{\Sigma}(q_l i_l(\tau)) - s_k^{(n+1)}(\tau) \underset{l}{\Sigma}(q_l i_l^{(n+1)}(\tau)) \big]\big| d\tau \nonumber \\
= & \int_0^t \big|\beta(\tau) k \big[ s_k(\tau)\underset{l}{\Sigma}(q_l i_l(\tau)) - s_k(\tau)\underset{l}{\Sigma}(q_l i_l^{(n+1)}(\tau)) \nonumber \\
& + s_k(\tau)\underset{l}{\Sigma}(q_l i_l^{(n+1)}(\tau)) - s_k^{(n+1)}(\tau) \underset{l}{\Sigma}(q_l i_l^{(n+1)}(\tau)) \big]\big| d\tau \nonumber \\
\leq & \beta_M k \underbrace{|s_k(\tau)|}_{\leq 1} \int_0^t q_M \underbrace{\big| \underset{l}{\Sigma} (i_l(\tau) - i_l^{(n+1)}(\tau)) \big|}_{\leq |\boldsymbol i(\tau) - \boldsymbol i^{(n+1)}(\tau)| } d\tau \nonumber \\
& + \beta_M k \underbrace{|\Sigma_l(q_l i_l^{(n+1)}(\tau))|}_{\leq 1} \int_0^t |s_k(\tau) - s_k^{(n+1)}(\tau)|d\tau. \label{eq:conv_mod_eik_1st_term}
\end{align}
Similarly, we estimate the second term of the RHS of (\ref{eq:conv_mod_eik}). Using it and (\ref{eq:conv_mod_eik_1st_term}) in (\ref{eq:conv_mod_eik}) we obtain:
\begin{align*}
& |e_{i_k}^{(n+1)}(t)| \leq \beta_{M} k \int_0^t q_{M} |\boldsymbol i(\tau) - \boldsymbol i^{(n+1)}(\tau)| d\tau \\
& + \beta_{M} k \int_0^t |i_k(\tau) - i_k^{(n+1)}(\tau)|d\tau \\
& + \gamma_{M}u_{M} \int_0^t |i_k(\tau) - i_k^{(n+1)}(\tau)|d\tau \\
& + \gamma_{M} \int_0^t |u_k(\tau) - u_k^{(n+1)}(\tau)|d\tau.
\end{align*}
Aggregating over all $k$ we obtain:
\begin{align}
& |\boldsymbol e_{i}^{(n+1)}(t)| \leq \big( \beta_{M}(\Sigma k) q_{M} + \beta_{M}K_{max} + \gamma_{M}u_{M} \big)  \nonumber \\
& ~~~~~~~~~~\times \int_0^t |\boldsymbol e_{i}^{(n+1)}(\tau)| d\tau + \gamma_{M}  \int_0^t |\boldsymbol e_{u}^{(n)}(\tau)|d\tau. \label{eq:err_i_wrt_i_u}
\end{align}

A similar procedure for $|\boldsymbol e_{\lambda}^{(n+1)}(t)|$ leads to (note that this error needs to be integrated backwards from $T$ to $t$):
\small
\begin{align}
& |\boldsymbol e_{\lambda}^{(n+1)}(t)| \leq \big( \beta_{M}(\Sigma k) \Lambda q_{M} + \beta_{M}\Lambda K_{max} \big) \int_t^T |\boldsymbol e_{i}^{(n+1)}(\tau)| d\tau  \nonumber \\
& ~~~~~~~ + \big( 2\beta_{M}K_{max} + \gamma_{M}u_{M} \big) \int_t^T |\boldsymbol e_{\lambda}^{(n+1)}(\tau)| d\tau \nonumber \\
& ~~~~~~~ + \gamma_{M}\Lambda  \int_t^T |\boldsymbol e_{u}^{(n)}(\tau)|d\tau. \label{eq:err_lam_wrt_i_lam_u}
\end{align}
\normalsize
Also, (\ref{eq:u_for_quard_cost}) leads to:
\begin{align}
|\boldsymbol e_{u}^{(n+1)}(t)| \leq \frac{\gamma_{M}\Lambda}{2c_{m}} |\boldsymbol e_{i}^{(n+1)}(t)| + \frac{\gamma_{M}}{2c_{m}} |\boldsymbol e_{\lambda}^{(n+1)}(t)|. \label{eq:err_u_wrt_i_lam}
\end{align}

The analysis uses following Gronwall's inequalities \cite{mcasey2012convergence}: if $\zeta,\kappa$ are two continuous functions on $[0,T]$ and $\kappa$ is non-decreasing, then
\begin{align}
\zeta(t) \leq \kappa(t) + \nu \int_0^t \zeta(\tau)d\tau \Rightarrow \zeta(t) \leq e^{\nu t}\kappa(t), \textnormal{ and,} \nonumber \\
\zeta(t) \leq \kappa(t) + \nu \int_t^T \zeta(\tau)d\tau \Rightarrow \zeta(t) \leq e^{\nu (T-t)}\kappa(t). \label{eq:gronwall}
\end{align}

Letting $c_0 = \beta_{M}(\Sigma k) q_{M} + \beta_{M}K_{max} + \gamma_{M}u_{M}$ and using first inequality of (\ref{eq:gronwall}) in (\ref{eq:err_i_wrt_i_u}) we obtain:
\begin{align}
& |\boldsymbol e_{i}^{(n+1)}(t)| \leq \exp(c_0 t) \gamma_{M} \int_0^t |\boldsymbol e_{u}^{(n)}(\tau)| d\tau \nonumber \\
\Rightarrow & |\boldsymbol e_{i}^{(n+1)}(t)| \leq \exp(c_0 t) \gamma_{M} \int_0^T |\boldsymbol e_{u}^{(n)}(\tau)| d\tau. \label{eq:err_i_wrt_u}
\end{align}
The estimation is true because integrand is positive. Again using second inequality of (\ref{eq:gronwall}) in (\ref{eq:err_lam_wrt_i_lam_u}) we get,
\begin{align}
& |\boldsymbol e_{\lambda}^{(n+1)}(t)| \leq \exp\{ (2\beta_{M}K_{max}+u_{M}\gamma_{M})(T-t) \} \nonumber \\
& \times \Big\{ \Lambda c_0 \int_t^T |\boldsymbol e_{i}^{(n+1)}(\tau)|d\tau + \gamma_{M}\Lambda \int_0^T |\boldsymbol e_{u}^{(n)}(\tau)|d\tau \Big\}. \label{eq:err_lam_wrt_i_u}
\end{align}
Using integration by parts in (\ref{eq:err_i_wrt_u}) to obtain $\int_t^T |\boldsymbol e_{i}^{(n+1)}(\tau)|d\tau$, and after ignoring some of the negative terms, (\ref{eq:err_lam_wrt_i_u}) can be estimated as:
\begin{align}
& |\boldsymbol e_{\lambda}^{(n+1)}(t)| \leq \exp\{ (2\beta_{M}K_{max}+u_{M}\gamma_{M})(T-t) \}\gamma_{M}\Lambda \nonumber \\
& ~~~~~~\times \Big\{ \exp(c_0T) \int_0^T |\boldsymbol e_{u}^{(n)}(\tau)|d\tau - \exp(c_0t) \int_0^t |\boldsymbol e_{u}^{(n)}(\tau)|d\tau \nonumber \\
& ~~~~~~~~~~~+ \cancel{\int_t^T (1-\exp(\underbrace{c_0\tau}_{\ge 0}))|\boldsymbol e_{u}^{(n)}(\tau)|d\tau}  \Big\}. \label{eq:err_lam_wrt_u}
\end{align}
The last term is negative, hence can be ignored given the direction of inequality. Note that, inspite of making estimations, the RHS of (\ref{eq:err_lam_wrt_u}) is maximum at $t=0$, and $0$ at $t=T$, as it should be (because transversaility condition fixes value of adjoint variables at $t=T$ leading to zero error at that point).

Using (\ref{eq:err_i_wrt_u}) and (\ref{eq:err_lam_wrt_u}) in (\ref{eq:err_u_wrt_i_lam}) we obtain
\begin{align}
& |\boldsymbol e_{u}^{(n+1)}(t)| \leq \frac{\gamma^2_{M}\Lambda}{2c_{m}} \exp(c_0t)\int_0^t \Big\{ 1-\exp((2\beta_{M}K_{max} \nonumber \\
& ~~~~~~~~~~~~~~~~+u_{M}\gamma_{M})(T-t)) \Big\} |\boldsymbol e_{u}^{(n)}(\tau)|d\tau \nonumber \\
& + \frac{\gamma^2_{M}\Lambda}{2c_{m}}\exp((2\beta_{M}K_{max}+u_{M}\gamma_{M})(T-t))\exp(c_0t) \nonumber \\
& \int_0^T |\boldsymbol e_{u}^{(n)}(\tau)| d\tau. \label{eq:err_u_wrt_u}
\end{align}

Putting $t=T$ in the first term of the RHS of (\ref{eq:err_u_wrt_u})---which gives an upper bound---and integrating from $0$ to $T$ gives
\begin{align*}
& |\boldsymbol e_{u}^{(n+1)}(t)| \leq \frac{\gamma^2_{M}\Lambda}{2c_{m}} \times \exp\{(\beta_{M}K_{max}+\gamma_{M}u_{M})T\} \times \\
& \Big[ \frac{\exp\{ \beta_{M}(\Sigma k) q_{M} T \} - \exp\{ \beta_{M} K_{max} T \} }{\beta_{M}(\Sigma k) q_{M} - \beta_{M} K_{max}} \Big] \int_0^T |\boldsymbol e_{u}^{(n)}(\tau)| d\tau.
\end{align*}
Thus, the algorithm converges when the leading constant is $<1$ which proves Theorem \ref{thm:conv_fw_bk_sweep}. Note that this is only a sufficient condition due to estimations made to obtain it.

Smaller values of $\gamma_M$ and larger values of $c_m$ aids convergence. We note that the function $(\exp(ax)-\exp(bx))/(ax-bx)$ is a monotonically increasing function of $x$ for $a,b>0$ and $a\neq b$. Thus, smaller values of $\beta_M$ aids convergence. Similarly, smaller $T$ leads to faster convergence.

\section{}
\label{app:proofs}
\textbf{\emph{Proof of Lemma \ref{thm:adjoint_variables_positive}:}} The Hamiltonian maximizing condition of the Pontryagin's Principle states that for all $t\in[0,T]$, $u_k^*(t) = \text{argmax}_{u_k(t)} H(\boldsymbol i^*(t),\boldsymbol \lambda^*(t),\boldsymbol u(t))$. Thus $\forall k\in\mathbb K$,
\begin{align}
& H(\boldsymbol i^*(t),\boldsymbol \lambda^*(t),u_{K_{min}}(t),...,u_k^*(t),...,u_{K_{max}}(t)) \nonumber \\
\geq & H(\boldsymbol i^*(t),\boldsymbol \lambda^*(t),u_{K_{min}}(t),...,0,...,u_{K_{max}}(t)). \nonumber
\end{align}
We use (\ref{eq:hamiltonian}) in the above. After simple algebraic manipulation and using $g_k(0)=0$, we get:
\begin{align}
\lambda_k^*(t)\gamma(t)u_k^*(t)s_k^*(t) \geq g_k(u_k^*(t)). \nonumber
\end{align}
By Assumption \ref{assumption:gk_increasing}, $g_k(u_k^*(t))\geq 0$. Notice also, $u_k^*(t)\geq 0$ (a consequence of Assumption \ref{assumption:gk_even}) and $s_k^*(t)\geq 0$ (from Lemma \ref{thm:ik_sk_in_0_1}) and $\gamma(t)\geq 0$. Hence $\lambda_k^*(t)\geq 0$.
\\

\textbf{\emph{Proof of Theorem \ref{thm:controls_structure}(i):}} Differentiating Eq. (\ref{eq:hamiltonian_max_cond_1}) with respect to time variable $t$, for any $k\in\mathbb K$, we get,
\begin{align}
g''_k(u_k^*(t))\frac{d}{dt} u_k^*(t) = \gamma(t) \frac{d}{dt} \lambda_k^*(t)~s_k^*(t) + \gamma(t) \lambda_k^*(t)\frac{d}{dt} s_k^*(t) \nonumber \\
+ \frac{d}{dt} \gamma(t) ~\lambda_k^*(t) s_k^*(t). \nonumber
\end{align}
Substituting the values of $\frac{d}{dt} \lambda_k^*(t)$ from Eq. (\ref{eq:costate_diff_eq}) and $\frac{d}{dt} s_k^*(t)$ from Eq. (\ref{eq:opt_prob_states}) ($s_k(t)=1-i_k(t)~\Rightarrow~\frac{d}{dt} s_k^*(t)=-\frac{d}{dt} i_k^*(t)$), and simplifying, we get,
\begin{align}
g''_k(u_k^*(t))\frac{d}{dt} u_k^*(t) = - \gamma(t) \beta(t) q_k s_k^*(t) \sum_{j\in\mathbb K} (j \lambda_j^*(t) s_j^*(t)) \nonumber \\
+ \frac{d}{dt} \gamma(t) ~\lambda_k^*(t) s_k^*(t). \label{eq:gdoubleprime_udot}
\end{align}
Now, $g''_k(u_k^*(t))\geq 0$ because $g_k(.)$ is assumed to be a convex function. The spreading rate $\beta(t)\geq 0$ and effectiveness of control $\gamma(t)\geq 0$ at all times, $q_k$ being a probability density function is $\geq 0$ \cite[Sec. 17.10.2]{newman2009networks}, $s_k^*(t)\geq 0,~\forall k\in\mathbb K$ (Lemma \ref{thm:ik_sk_in_0_1}) and  $\lambda_k^*(t)\geq 0,~\forall k\in\mathbb K$ (Lemma \ref{thm:adjoint_variables_positive}). For the second term $\frac{d}{dt}\gamma(t)\leq 0$. Hence we get $\frac{d}{dt} u_k^*(t)\leq 0,~\forall k\in \mathbb K$, which proves the theorem.
\\

\textbf{\emph{Proof of Theorem \ref{thm:controls_structure}(ii):}} For $g_k(u_k(t))=c_ku_k^2(t)$, Eq. (\ref{eq:gdoubleprime_udot}) can be re-written as $2c_k \frac{d}{dt} u_k^*(t)=-\beta(t) q_k s_k^*(t) \sum_{j\in\mathbb K} (j g_j'(u_j^*(t)))+\frac{d}{dt} \gamma(t) \lambda_k^*(t) s_k^*(t)$ (using Eq. (\ref{eq:hamiltonian_max_cond_1})). This, on differentiating with respect to $t$, leads to,
\small
\begin{align}
2c_k\frac{d^2}{dt^2}{u}_k^*(t)=&-\underbrace{\beta(t) q_k}_{\geq 0} \underbrace{\frac{d}{dt} s_k^*(t)}_{\leq 0\text{ from }(\ref{eq:opt_prob_states})} \underbrace{\sum_{j\in\mathbb K} (j g_j'(u_j^*(t)))}_{\geq 0\text{ from Assumption \ref{assumption:gk_increasing}}} \nonumber \\
& -\underbrace{\beta(t) q_k s_k^*(t)}_{\geq 0} \sum_{j\in\mathbb K} \Big( j \underbrace{g_j''(u_j^*(t))}_{=2c_j>0} . \underbrace{\frac{d}{dt} u_j^*(t)}_{\leq 0\text{ Theorem \ref{thm:controls_structure}(i)}}\Big) \nonumber \\
& - \underbrace{\frac{d}{dt}{\beta}(t)}_{\leq 0}q_k \underbrace{s_k^*(t)}_{\geq 0\text{ Lemma \ref{thm:ik_sk_in_0_1}}} \underbrace{\sum_{j\in\mathbb K} (j g_j'(u_j^*(t)))}_{\geq 0\text{ from Assumption \ref{assumption:gk_increasing}}} \nonumber \\
& + \underbrace{\frac{d}{dt} \gamma(t)}_{\leq 0} ~ \underbrace{\Big( -\beta(t) q_k s_k^*(t) \sum_{j\in\mathbb K} (j \lambda_k^*(t)s_k^*(t)) \Big)}_{\leq 0} \nonumber \\
& + \underbrace{\frac{d^2}{dt^2}{\gamma}(t)}_{\geq 0} ~\underbrace{\lambda_k^*(t) s_k^*(t)}_{\geq 0}.\nonumber
\end{align}
\normalsize
This leads to the conclusion that $\overset{..}{u}_k^*(t) \geq 0$, $\forall k\in\mathbb K$ and $t\in[0,T]$, for $\frac{d}{dt}{\beta}(t),\frac{d}{dt} \gamma(t)\leq 0$ and $\frac{d^2}{dt^2}{\gamma}(t) \geq 0$.

\section{Proof of Theorem \ref{thm:uniqueness_of_state_adjoint_opt_control}:Uniqueness of Solutions}
\label{app:proof_uniqueness}

Condition (\ref{eq:cond_for_uniqueness}) of Theorem \ref{thm:uniqueness_of_state_adjoint_opt_control} is obtained by applying the theorem for uniqueness of the solution of first order two point boundary value problem in \cite{ma2002existence}:

\begin{lemma}[Theorem 1.2 in \cite{ma2002existence}]
\label{thm:uniq_2pt_bvp_soln}
Let $\det (M_{2n\times 2n}+R_{2n\times 2n})\neq 0$, $\alpha\in \mathbb R^{2n}$ and $\boldsymbol \psi :[0,T]\times \mathbb R^{2n}\rightarrow R^{2n}$ be a Carath\'eodary function. If $\exists w(t)\in L^1([0,T])$ such that $||\boldsymbol \psi(t,\boldsymbol \eta)-\boldsymbol\psi(t,\boldsymbol{\hat \eta})||\leq w(t)||\boldsymbol\eta-\boldsymbol{\hat \eta}||$ $\forall t\in [0,T], \boldsymbol\eta,\boldsymbol{\hat \eta}\in \mathbb R^{2n}$. Then the two point first order boundary value problem $\frac{d}{dt}\boldsymbol\eta=\boldsymbol \psi(t,\boldsymbol\eta);~M\boldsymbol\eta(0)+R\boldsymbol\eta(T)=\alpha$ has a unique solution provided $\Gamma||w||_{L_1}<1$. Here, $\Gamma=\max\{ ||(M+R)^{-1}M||, ||(M+R)^{-1}R|| \}$ and for a matrix, $||X||=\max_{i,j}|X_{i,j}|$.
\end{lemma}

Note that in our case, $n=|\mathbb K|$, $\boldsymbol \eta = (\boldsymbol i, \boldsymbol \lambda)$, $\boldsymbol\psi = (\boldsymbol f, \boldsymbol h)$, where $f_k(t,\boldsymbol \eta)$ is RHS of (\ref{eq:opt_prob_states}) and $h_k(t,\boldsymbol \eta)$ is RHS of (\ref{eq:costate_diff_eq}) for $k\in \mathbb K$. $M$ and $R$ are both $2n\times 2n$ matrices such that $M_{jj}=1=R_{j+n,j+n},1\leq j\leq n$ and all other elements are $0$. Thus, $\Gamma=1$.

We note that $\boldsymbol \psi$ is a Carath\'eodary function because it satisfies the following requirements:
\begin{enumerate}
\item $\boldsymbol \psi(.,\boldsymbol \eta)$ is Lebesgue measurable on $[0,T]~\forall \boldsymbol \eta \in \mathbb R^{2n}$.
\item $\boldsymbol \psi(t,.)$ is continuous on $\mathbb R^{2n}~\forall t\in[0,T]$.
\item For all $r\in(0,\infty), ~t\in[0,T], ~||\boldsymbol \eta|| \leq r, \exists~\delta_r:[0,T]\times \mathbb R^{2n}\rightarrow \mathbb R^{2n}$ such that $|\psi_k(t,\boldsymbol \eta)|\leq (\delta_r)_k(t), ~1\leq k\leq 2n$. This is satisfied for $(\delta_r)_k(t)=\beta(t) k q_{M} r + \frac{\gamma^2(t)}{2c_k}r$ for $1\leq k\leq n$, and $(\delta_r)_k(t) = \beta(t) K_{max} q_{M} r^2 + \frac{\gamma^2(t)}{2c_k}r^2$  for $n+1\leq k\leq 2n$.
\end{enumerate}

Noting, $g_k(u_k(t))=c_ku^2_k(t) \Rightarrow u_k(t)=\frac{\gamma(t)\lambda_k(t)s_k(t)}{2c_k}$ (from (\ref{eq:hamiltonian_max_cond_1})) and letting $\boldsymbol {\hat\eta} = (\boldsymbol {\hat i}, \boldsymbol {\hat \lambda})$; after some algebraic manipulations, for $1\leq k\leq n$, we obtain:
\begin{align*}
|f_k(t,\boldsymbol \eta)-f_k(t,\boldsymbol{\hat \eta})| \leq \beta(t) k q_{M} ||\boldsymbol i - \boldsymbol{\hat i}|| + \beta(t) k |s_k-\hat s_k| \\
+ \frac{\gamma^2(t)}{2c_k} \Lambda |s_k-\hat s_k|\times 2 + \frac{\gamma^2(t)}{2c_k} |\lambda_k-\hat\lambda_k|, \\
|h_k(t,\boldsymbol \eta)-h_k(t,\boldsymbol{\hat \eta})| \leq \beta(t) k \Lambda q_{M} ||\boldsymbol i - \boldsymbol{\hat i}|| + \beta(t) k |\lambda_k-\hat\lambda_k| \\
+ \beta(t) q_k K_{max} (\Lambda ||\boldsymbol i - \boldsymbol{\hat i}|| + ||\boldsymbol \lambda - \boldsymbol{\hat \lambda}||) \\
+ \frac{\gamma^2(t)}{2c_k} (\Lambda^2|s_k-\hat s_k| + 2\Lambda|\lambda_k-\hat\lambda_k|).
\end{align*}
Aggregating over all $k$ and noting that $\boldsymbol\psi = (\boldsymbol f, \boldsymbol h)$,
\begin{align*}
& ||\boldsymbol \psi(t,\boldsymbol \eta) - \boldsymbol \psi(t, \boldsymbol{\hat \eta})|| = ||\boldsymbol f(t,\boldsymbol \eta) - \boldsymbol f(t,\boldsymbol{\hat \eta})|| \\
& ~~~~~~~~~~~~~~~~~~~~~~~~~~~ + ||\boldsymbol h(t,\boldsymbol \eta) - \boldsymbol h(t,\boldsymbol{\hat \eta})|| \\
\leq & ||\boldsymbol i - \boldsymbol{\hat i}|| \Big( \beta(t)(\Sigma k) q_{M} + \beta(t)(\Sigma k) + \frac{\gamma^2(t)}{c_{m}}\Lambda \Big)  + \\
& ||\boldsymbol \lambda - \boldsymbol{\hat \lambda}|| \frac{\gamma^2(t)}{2c_{m}} + ||\boldsymbol i - \boldsymbol{\hat i}|| \Big( \beta(t)(\Sigma k) q_{M} \Lambda + \beta(t)K_{max}\Lambda \\
& + \frac{\gamma^2(t)}{2c_{m}}\Lambda^2 \Big) + ||\boldsymbol \lambda - \boldsymbol{\hat \lambda}|| \Big( 2\beta(t)K_{max} + \frac{\gamma^2(t)}{c_{m}}\Lambda \Big) \\
= & [d_1 \beta(t) + d_2 \gamma^2(t)]~ ||\boldsymbol\eta-\boldsymbol{\hat \eta}||,
\end{align*}
where,  $d_1=\max\{ (\sum_{k\in\mathbb K}k)\Lambda q_{M} + K_{max}\Lambda, 2K_{max} \},~d_2=(\Lambda/c_{m})\max\{ 1,\Lambda/2 \}$.

As stated above, for present case $\Gamma=1$; thus, the solutions to the state and adjoint equations (and hence the controls) are unique when $d_1||\beta(t)||_{L_1}+d_2||\gamma^2(t)||_{L_1}<1$ (using Lemma \ref{thm:uniq_2pt_bvp_soln}).

\end{document}